\documentclass[journal]{IEEEtran}
\usepackage{amsfonts,amsmath,amssymb}
\usepackage{dsfont}
\usepackage{cite}
\usepackage{graphicx}
 \usepackage{booktabs}
\usepackage{url}
 \usepackage{caption}
\usepackage{subcaption}
\usepackage{bm}
\usepackage{bbding}
\usepackage{amsmath}
\usepackage{enumerate}
\usepackage{multirow}
\usepackage{algorithm}
\usepackage{cancel}
\usepackage{xcolor}%
\usepackage{algcompatible}
\usepackage{soul}
\usepackage{pdfpages}
\usepackage{epstopdf}
\usepackage{slashbox}

\usepackage{array}
\usepackage{CJK}
\usepackage{graphicx}
\usepackage{amsthm}

  \newtheorem{proposition}{Proposition}
 \newtheorem{corollary}{Corollary}
 
  \newtheorem{remark}{Remark}

\newcommand{\figref}[1]{\figurename~\ref{#1}}

\DeclareMathOperator{\diag}{diag}

\begin{document}
\title{On the Spectrum of OTFS/VOFDM Signals: PSD Analysis and Bandwidth Allocation}
\author{Wei~Wang,~\IEEEmembership{Senior Member,~IEEE}, and~Xiang-Gen~Xia,~\IEEEmembership{Fellow,~IEEE}
\thanks{  W. Wang is  with the School of Information Science and Technology, Harbin Institute of Technology (Shenzhen), Shenzhen
518055, China (e-mail: wang\_wei@hit.edu.cn).

X.-G. Xia is with the Department of Electrical and Computer Engineering, University of Delaware, Newark, DE 19716, USA (e-mail: xxia@ece.udel.edu).
}
}

\maketitle

\begin{abstract}
Orthogonal time frequency space (OTFS)/vector OFDM (VOFDM) is widely regarded as a promising waveform for next-generation mobile communications. However, its spectral characteristics are not yet fully understood. The bandwidth allocation scheme, which is crucial for OTFS's integration into practical wireless standards, also remains unexplored. In this paper, we investigate the spectral characteristics of OTFS signals by analyzing their power spectral density (PSD). We demonstrate that the PSD of discrete-time OTFS signals is periodic with a period of $\frac{1}{MT_s}$, where $M$ is  the size of the time/Doppler domain in OTFS, a.k.a., the vector size in VOFDM, and $T_s$ is the sampling interval length of digital to analog converter (DAC), resulting in $M$ identical spectral components within the spectral range $[-\frac{1}{2T_s}, \frac{1}{2T_s})$ of the continuous-time OTFS signal. The periodicity makes bandwidth allocation for OTFS/VOFDM signaling substantially challenging. Furthermore, we establish a relationship between the PSD of OFDM signals and that of OTFS signals, revealing that, when the information symbols are independent, the PSD of OTFS signals is equal to the sum of the PSDs of the component-expanded OFDM (CEP-OFDM) signals. Lastly, we derive a relationship between the information symbols and the corresponding OTFS spectrum, and based on which, we propose a null-space-based linear precoding (NSLP) method for OTFS signals to enable flexible bandwidth allocation. Numerical results validate our analytical results regarding the PSD of OTFS signals and show the effectiveness of our proposed NSLP method in tailoring the spectrum of OTFS signals.
\end{abstract}


\section{Introduction}

Every transition to a new generation of mobile communications involves a disruption in the underlying air interface \cite{hadani2018otfs, OTFS, liyanaarachchi2021optimized, 5Gwaveform}. The 2G Global System for Mobile Communications (GSM) is characterized by time-division multiple access (TDMA), along with the use of  Gaussian filtered minimum shift keying (GMSK) \cite{turletti1996gmsk, redl1995introduction}, IMT-2000 (3G) is characterized by code-division multiple access (CDMA) \cite{dahlman1998umts},  IMT-Advanced (4G) is characterized by orthogonal frequency division multiplexing (OFDM) \cite{dahlman20134g}, and IMT-2020 (5G) continues adopting OFDM \cite{zaidi2016waveform}. The interruption of development in modulation technology from 4G to 5G has hindered the performance breakthrough in physical layer.  

Modulation techniques tailored for delay–Doppler channels (i.e., doubly selective channels) have recently regained significant research interest, driven by advances in LEO satellite communications, UAV networks, and intelligent transportation systems \cite{xia2025rethink}.  With new capabilities, e.g., high resilience to Doppler and delay spread, and improved robustness to multipath channels, orthogonal time frequency space (OTFS) is widely regarded as a promising waveform for IMT-2030 (6G) communications \cite{das2022orthogonal,UnifyingOTFS, wei2021orthogonal, xiao2021overview}. However, it is important to note that OTFS is in fact equivalent \cite{raviteja2019otfs, van2024equivalence,PCCHINGOTFS, xia2022comments} to the earlier modulation schemes, such as vector OFDM (VOFDM) \cite{xia2000precoded, xia2001precoded}, orthogonal signal-division multiplexing (OSDM) \cite{SuehiroKronecker}, and asymmetric OFDM (A-OFDM) \cite{zhang2007asymmetric}. Therefore, throughout this paper, we will use OTFS as a unified term for the modulation schemes discussed above.

The information symbols of OTFS, i.e., QAM/PSK modulated symbols, are embedded in the delay-Doppler domain \cite{hadani2018otfs}. However, wireless signals are regulated in their spectrum.  For a wireless communication system, the signal bandwidth is basically determined by the air interface technologies, which refer to the digital signal processing algorithms that  map information bits to the waveform \cite{jondral2005software, SDRreview}. For the prevailing technology OFDM/OFDMA, which is extensively used in various wireless standards, e.g., IEEE 802.11,  DVB-T, DVB-H, 4G-LTE, and 5G-NR, the information stream is firstly converted to a parallel stream, which are interpreted as symbols in the frequency domain, and then transformed into the discrete-time samples by passing through an inverse fast Fourier transform (IFFT) module. The over-the-air signal is subsequently generated by the digital-to-analog converter (DAC) and power amplifier \cite{SDRofdm}. The spectrum of the generated waveform is determined by 1) the mapping from information  symbols to the discrete-time waveform sample, 2) the digital shaping filter and the interpolation filter of DAC, and 3) the sampling rate of the DAC. To control the bandwidth, OFDM can simply transmit zeros on unwanted subcarriers \cite{van2010analytical, rajabzadeh2017power,talbot2008spectral}. In a word, the strategy for bandwidth allocation in OFDM is both straightforward and effective.

Unlike OFDM signals, where the correspondence between the information symbols and the spectrum is clearly elaborated via the concept of subcarrier, the relationship between a single QAM/PSK symbol and its  resulting spectrum of OTFS signal remains almost unknown. Understanding the relationship between the information symbols and the spectrum of the OTFS waveform is crucial for bandwidth allocation, making it essential for the practical applications of OTFS.

Power spectral density (PSD), which is a measure of the power distribution across frequencies, is an important spectral property of the signal \cite{youngworth2005overview}.  PSD analysis is vital as the spectral properties of  wireless equipments must comply with the regulations and standards about electromagnetic emissions established by the regulatory bodies, such as Federal Communications Commission (FCC) of the U.S. and State Radio Regulatory Commission (SRRC) of China. The PSD for OTFS signaling is numerically investigated in \cite{UnifyingOTFS} without considering bandwidth allocation. To explore the rationale behind how zero-settings in the information-bearing matrix $\mathbf{X}$ affect the PSD of OTFS signaling, we first provide a theoretical investigation of the PSD characteristics of OTFS. Building on this analysis, we further propose an efficient and effective bandwidth allocation scheme, thereby extending our previous work on the discrete spectrum analysis of OTFS/VOFDM signals in \cite{xia2024discrete}.

To summarize, the contributions of this paper are as follows.
\begin{itemize}
\item We derive the analytical expression of  PSD of  OTFS signals generated by DAC. In addition, we reveal the fact that the PSD of discrete-time OTFS signals is periodic with a period of $\frac{1}{MT_s}$, where $M$ is the number of component-expanded OFDM (CEP-OFDM) signals comprising the OTFS signal and $T_s$ is the sampling interval length of DAC, resulting in $M$ identical spectral components within the spectral range $[-\frac{1}{2T_s}, \frac{1}{2T_s})$ of the continuous-time OTFS signal generated by DAC.
\item We establish the relationship between the PSD of OFDM signals and that of OTFS signals. Specifically, when the information symbols are independent, the PSD of OTFS signals is the sum of the PSDs of the CEP-OFDM signals that form the OTFS signal. This further connects OFDM and OTFS in the spectrum domain. It illustrates that the OTFS signal is a time-interleaved version of $M$ OFDM signals, each with an FFT/IFFT size of $N$, and can be obviously seen from  VOFDM \cite{xia2000precoded, xia2001precoded}. 
\item We derive the relationship between the information symbols and the corresponding spectrum of OTFS signals. Based on the relationship, we propose a null-space-based linear precoding (NSLP) method for OTFS signals that enables flexible bandwidth allocation. Based on our derived expression of NSLP precoder, we further obtain the expression of the systematic-form precoder.
\end{itemize}
Numerical simulations validate our analytical results regarding the PSD of OTFS signals and show the effectiveness of our proposed NSLP method in tailoring the spectrum of OTFS signals.

{\em{Notations:\quad}} Column vectors (matrices) are denoted by bold-face lower (upper) case letters, $\mathbf{x}[n]$ denotes the $n$-th element in the vector $\mathbf{x}$, $(\cdot)^*$, $(\cdot)^T$ and $(\cdot)^{H}$  represent conjugate, transpose and conjugate transpose operation, respectively,  $\otimes$ denotes Kronecker product, $\mathcal{F}(\cdot)$ denotes the Fourier transform, and $\diag(x_1, x_2,\cdots, x_N )$  represents a diagonal matrix with the given elements $x_1, x_2,\cdots, x_N$ on its main diagonal, $\| \mathbf{x} \|_2$ is  $\ell^2$ norm of the vector $\mathbf{x}$,  $\mathbb{E}(\cdot)$  represents the expectation of a random variable, and $|\mathcal{I}|$ is the cardinality of the set $\mathcal{I}$.


\section{PSD Analysis of OTFS Signals Generated by DAC}

Although many existing analytical results of the multi-carrier systems, e.g., OFDM and OTFS, are based on an analog representation of the signal, the actual baseband signal is first constructed in the digital domain by the digital signal processing (DSP) module FFT/IFFT, and then converted to the analog signal via a DAC \cite{lin2003analog}. Therefore,  in this section, we analyze the PSD of OTFS signals generated by DAC.

\subsection{The Waveform of OTFS Signals Generated by DAC}

The bit-to-symbol mapping of OTFS is carried out in the delay-Doppler domain, and then an Inverse Symplectic Fast Fourier Transform (ISFFT) precoding is performed to convert the QAM/PSK signal  $x_{l,k}$ to the time-frequency signal $\widetilde{x}_{m,n}$. Subsequently,  $\widetilde{x}_{m,n}$ is turned into the time-domain waveform by an OFDM modulator using  the Heisenberg transform that consists of IFFT operation and pulse-shaping operation.

For the $i$-th OTFS symbol, the information symbol matrix  $ {\mathbf{X}}_i \in \mathbb{C}^{M \times N}$ is in the delay-Doppler domain. After ISFFT precoding, the QAM/PSK signal in delay-Doppler domain will be transformed into the time-frequency domain as follows
\begin{align}
\widetilde{x}_{i, m, n} =  \frac{1}{\sqrt{MN}} \sum_{l=0}^{M-1} \sum_{k=0}^{N-1}  {x}_{i,l,k}   e^{\jmath 2\pi (\frac{k n}{N}-\frac{l m}{M})}
\end{align}
where $ {x}_{i,l,k}$ is the $(l, k)$-th element of $\mathbf{X}_i$ that denotes the QAM/PSK constellation signal of the $(l,k)$-th bin  in the delay-Doppler domain of the $i$-th OTFS symbol.

Then, the discrete-time OTFS signal ${s}_{i,l,n}$ is obtained by further transforming $ \widetilde{x}_{i, m, n} $ to the time domain as follows
\begin{align}
 {s}_{i,l,n} & = \frac{1}{\sqrt{M}}\sum_{m=0}^{M-1}   \widetilde{x}_{i, m, n} e^{\jmath  {2\pi}\frac{l m}{M}}   g_P[l] \notag \\
& = \frac{1}{\sqrt{N}}  \sum_{k=0}^{N-1}  {x}_{i, l,k}   e^{\jmath 2\pi \frac{k n}{N} }     g_P[l] \label{OTFS1}
\end{align}
where ${g}_P[l]$ is the digital shaping filter\footnotemark. Note that ${s}_{i,l,n}$ is the discrete-time OTFS signal before the cyclic prefix (CP) insertion. Since this paper is only interested in the PSD of OTFS signals, for simplicity, we do not consider a CP.

\footnotetext{  In this paper, we further set $g_P[l] = 1$ by incorporating the impacts of $g_P[l]$ into the signal constellation for $x_{i,l,k}$. Consequently, we set $\mathbf{G}_P = \diag(\mathbf{g}_P) = \mathbf{I}_M$ in the subsequent context.}

According to \cite{DiverOTFS, PracOTFS,GeOTFS}, the matrix form of \eqref{OTFS1} is represented by
\begin{align}
 {\mathbf{S}}_i = \mathbf{I}_M  {\mathbf{X}}_i \mathbf{F}_N^H  \label{Smatrix}
\end{align}
where ${\mathbf{S}}_i  \in \mathbb{C}^{M\times N}$ and $\mathbf{F}_N \in \mathbb{C}^{N\times N}$ is the DFT matrix.
By column-wise vectorizing $ {\mathbf{S}}_i$, the vector-form of the discrete baseband signal is given by
\begin{align}
 {\mathbf{s}}_i = {\rm{vec}}( {\mathbf{S}}_i ) = (\mathbf{F}_N^H \otimes \mathbf{I}_M) {\rm{vec}}( {\mathbf{X}}_i) \label{OTFSdiscrete}
\end{align}
In this regard, the $\eta$-th element of the discrete-time OTFS signal is given by
\begin{align}
s_{\eta}  =  \sum_{i = -\infty}^{\infty} \sum_{n = 0}^{N-1} \sum_{l=0}^{M-1} {s}_{i,l,n}  \delta [\eta - iMN - nM - l]  \label{Stime}
\end{align}
where ${s}_{i,l,n}$ is the $(l,n)$-th entry of the matrix $\mathbf{S}_i$, as well as the $(nM + l)$-th entry of the vector $\mathbf{s}_i$.

In the following proposition, we derive the cyclostationary property for the discrete-time OTFS signal $s_{\eta}$.
\begin{proposition} {\rm
When the information  symbols $x_{i,l,k}$ are i.i.d. with zero mean and  variance $\mathbb{E}\left( {x}^*_{i,l,k} {x}_{{i}, {l}, {k}} \right) = \sigma^2_{l,k}$ with respect to the index $i$, and are independent across the indices $(l,k)$, the resulting discrete-time OTFS signal  ${s}_{\eta}$ in \eqref{Stime} is cyclostationary with a period of $MN$.  }
\end{proposition}
\begin{proof}
See Appendix A.
\end{proof}

Feed the discrete-time baseband signal $s_{\eta}$ to a DAC,  the analog signal is obtained as
\begin{align}
  {s}(t)  =  &\sum_{\eta= -\infty}^{\infty}  {s}_{\eta} g_I( t- \eta T_s ) \notag \\
=  &\sum_{i= -\infty}^{\infty} \sum_{n= 0}^{N-1} \sum_{l=0}^{M-1}  {s}_{i,l,n} g_I( t- (iNM + nM + l) T_s ) \notag \\
=  &\sum_{i= -\infty}^{\infty} \sum_{n= 0}^{N-1} \sum_{l=0}^{M-1}  \sum_{k=0}^{N-1}  \frac{{x}_{i,l,k}}{\sqrt{N}}    e^{\jmath 2\pi \frac{k n}{N} }  \cdot \notag \\
&  \qquad \qquad \qquad g_I( t - (iNM + nM + l) T_s) \label{OTFSbaseband}
\end{align}
where $g_I(t)$ is the continuous-time interpolation filter, $T_s$ is the sampling interval length and $f_s = 1/{T_s}$ is the sampling frequency. The procedures of generating the continuous-time OTFS signal are given in \figref{VOFDMwaveform}.

\begin{figure}[tp]{
\begin{center}{\includegraphics[width=7cm ]{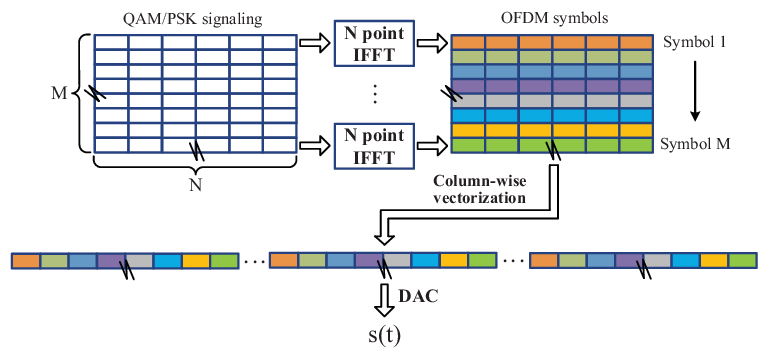}}
\caption{The procedures of generating the continuous-time OTFS signal}\label{VOFDMwaveform}
\end{center}}
\end{figure}

\subsection{The PSD of OTFS Baseband Signal Generated by DAC}

In the following proposition, we derive the PSD  for the continuous-time OTFS signal $s(t)$.
\begin{proposition} {\rm
The resulting continuous-time OTFS signal  $s(t)$ in \eqref{OTFSbaseband} under the conditions in Prop. 1 is cyclostationary with a period of $MNT_s$, and the PSD of the continuous-time OTFS signal $s(t)$  is then given by
\begin{align}
P_c(f) = \underbrace{\sum_{k=0}^{N-1}  \frac{\sigma_{k}^2}{T_s}  \left( \frac{ \sin(\pi(k-fMNT_s) ) }{N \sin(\frac{ \pi}{N} (k-fMNT_s) ) } \right)^2}_{P_d(f)}  |G_I(f)|^2   \label{PSD}
\end{align}
where  $\sigma^2_{k}  = \sum_{l=0}^{M-1}  \frac{\sigma^2_{l,k}}{M}$, $ G_I(f) = \mathcal{F}(g_I(t)) $ is the Fourier transform of  the interpolation filter $g_I(t)$,  and $P_d(f)$ is the PSD of the discrete-time OTFS signal $s_{\eta} $ in \eqref{Stime}  with the sampling rate $\frac{1}{T_s}$.
}
\end{proposition}

\begin{proof}
See Appendix B.
\end{proof}

Note that although the PSD in \eqref{PSD} for the continuous-time signal $s(t)$ is the same as that of the conventional one in \cite{proakis2008digital} as the product of the PSD of the discrete-time signal and the interpolation pulse $g_I(t)$,  its proof is different and we cannot simply apply that in \cite{proakis2008digital} . This is because in \cite{proakis2008digital} , the discrete-time signal needs to be wide sense stationary (WSS), for example when the information symbol sequence $x_{i,l,k}$ is assumed  i.i.d. across all three indices $i$, $l$, and $k$.   Since in our paper, we are interested in bandwidth allocation for OTFS signals, the information symbol sequence $x_{i,l,k}$ may not be i.i.d. across the indices $l$ and $k$ for any fixed OTFS symbol index $i$.  Thus, the discrete-time signal $s_\eta$ in \eqref{Stime} is not WSS in general and is only cyclostationary.

Also note that the term $  \frac{ \sin(\pi(k-fMNT_s) ) }{N \sin(\frac{ \pi}{N} (k-fMNT_s) ) } $  in Eq. \eqref{PSD} is a Dirichlet function, or periodic/aliased sinc function \cite{fessler2003nonuniform}.
Thus, Eq. \eqref{PSD} is, in essence, the interpolation of the discrete sequence $\{ \frac{\sigma_{k}^2}{T_s} \}$ with the squared Dirichlet kernel, i.e.,
\begin{align}
D^2_N(k-fMNT_s) =   \left(\frac{ \sin(\pi(k-fMNT_s) ) }{N \sin(\frac{ \pi}{N} (k-fMNT_s) ) } \right)^2  \label{sqDir}
\end{align}
under the windowing effects of the interpolation filter $g_I(t)$, whose Fourier transform is $G_I(f)$. The squared Dirichlet kernel
\begin{align}
    D^2_N(x) \triangleq \frac{\sin^2(\pi x)}{N^2 \sin^2(\frac{\pi x}{N})} \label{Dir1}
\end{align}
is a periodic function with the period $N$, and when $x=iN, i\in\mathbb{Z}$, the squared Dirichlet kernel reaches its peak $D^2_N(x) = 1$. Thus, the term \eqref{sqDir}  is a periodic function with the period $\frac{1}{MT_s}$, and $D^2_N(k-fMNT_s) = 1 $, when $f =\frac{k-iN}{MNT_s} $.

Further, we set $k=0$ and plot the squared Dirichlet kernel $D^2_N(fMNT_s)$ with different parameters to study its shape in \figref{DirichletPlot}.
\begin{itemize}
\item When $M=1, N=16$, $D^2_{N=16}(16fT_s)$ is periodic with the period $\frac{1}{MT_s} = \frac{1}{T_s}$ and the width of the main lobe, a.k.a., spectral spike, is $\frac{2}{MNT_s} = \frac{1}{8T_s}$. Note that, in this case,  there is only $M=1$ spike within the range $[-1/2T_s, 1/2T_s)$. OTFS turns to be OFDM, and the squared Dirichlet kernel corresponds to the spectrum of a subcarrier. Recall \eqref{PSD}, the weight of the squared Dirichlet kernel $D^2_{N=16}(16fT_s)$ that forms PSD is determined by
     \begin{align}
        \sigma_{0}^2  =  \lim_{I \rightarrow \infty} \frac{\sum_{i= 0}^{I-1} |x_{i, 0, 0}|^2}{I} \notag
      \end{align}
     which means when the information signal $x_{i, 0, 0}$ is set to $0$, the spectral range
     \begin{align}
     \bigcup\limits_{\ell \in \mathbb{Z}}  [-\frac{1}{16T_s} +  \ell \frac{1}{T_s}, \frac{1}{16T_s}+\ell \frac{1}{T_s}  ) \notag
     \end{align}
      remains unoccupied.

\item When $M=2, N=8$, $D^2_{N=8}(16fT_s)$ is periodic with the period $\frac{1}{MT_s} = \frac{1}{2T_s}$ and the width of the main lobe is $\frac{2}{MNT_s} = \frac{1}{8T_s}$. In this case,  there are $M=2$ spikes within the range $[-1/2T_s, 1/2T_s)$. Recall \eqref{PSD}, the weight of the squared Dirichlet kernel $D^2_{N=8}(16fT_s)$ that forms PSD is determined by
     \begin{align}
        \sigma_{0}^2  =  \lim_{I \rightarrow \infty} \frac{\sum_{i= 0}^{I-1} \left(|x_{i, 0, 0}|^2+|x_{i, 1, 0}|^2\right)} {2I} \notag
      \end{align}
     which means when the information signals $x_{i, 0, 0}, x_{i, 1, 0}$ are both set to $0$, the spectral range
    \begin{align}
        \bigcup\limits_{\ell \in \mathbb{Z}}  [-\frac{1}{16T_s} +  \ell \frac{1}{2T_s}, \frac{1}{16T_s}+\ell \frac{1}{2T_s}  ) \notag
    \end{align}
  remains unoccupied.

\begin{figure}[tp]{
\begin{center}{\includegraphics[width=7cm ]{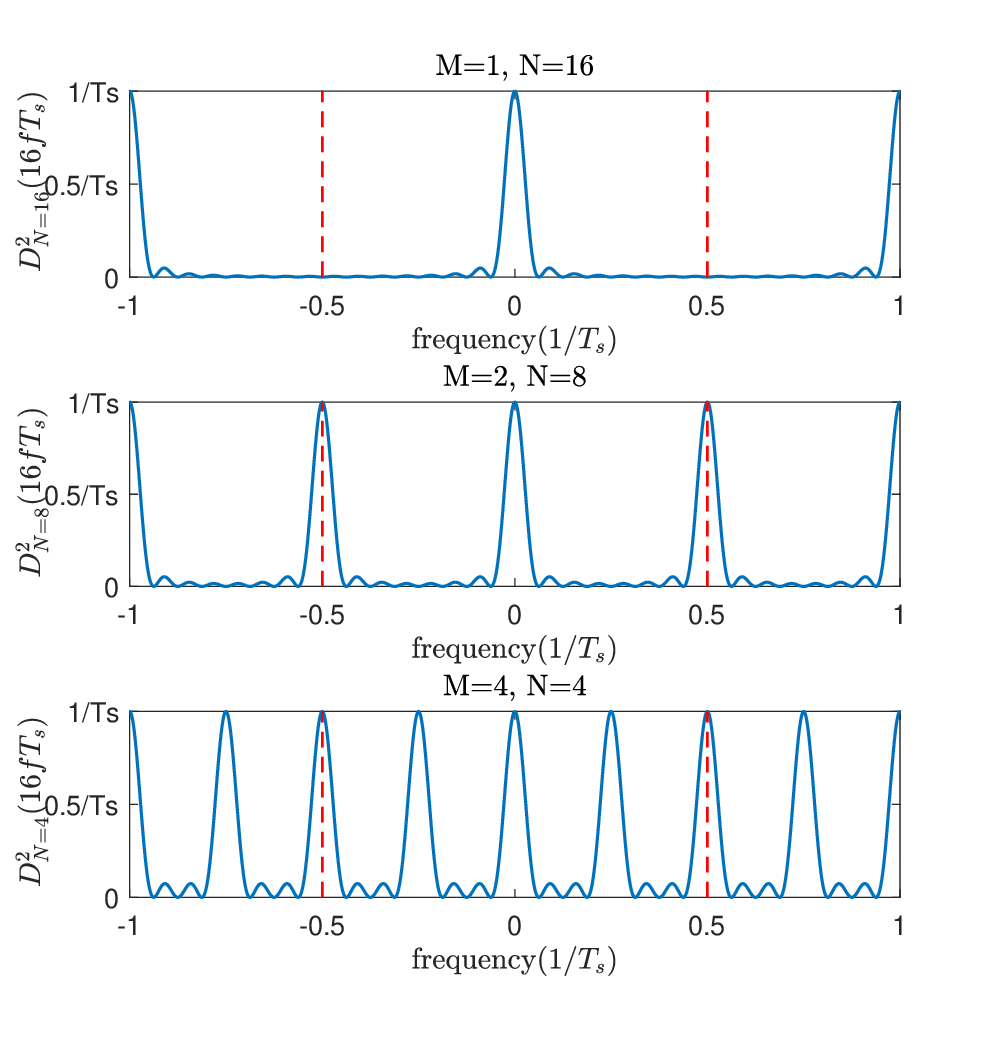}}
\caption{Plot of the squared Dirichlet kernel $D^2_N(fMNT_s)$ with different parameters. (Note that there are $M$ spectral spike/spikes within the range $[-1/2T_s, 1/2T_s)$.) }\label{DirichletPlot}
\end{center}}
\end{figure}

\item When $M=4, N=4$, $D^2_{N=4}(16fT_s)$ is periodic with the period $\frac{1}{MT_s} = \frac{1}{4T_s}$ and the width of the main lobe is $\frac{2}{MNT_s} = \frac{1}{8T_s}$.   In this case,  there are $M=4$ spikes within the range $[-1/2T_s, 1/2T_s)$. Recall \eqref{PSD}, the weight of the squared Dirichlet kernel $D^2_{N=4}(16fT_s)$ that forms PSD is determined by
     \begin{align}
     &\sigma_{0}^2 \notag \\
      =  &\lim_{I \rightarrow \infty} \frac{\sum_{i= 0}^{I-1} \left(|x_{i, 0, 0}|^2+|x_{i, 1, 0}|^2+|x_{i, 2, 0}|^2+|x_{i, 3, 0}|^2\right)} {4I} \notag
      \end{align}
     which means when the information signals $x_{i, 0, 0}, x_{i, 1, 0}, x_{i, 2, 0}, x_{i, 3, 0}$ are all set to $0$, the spectral range
    \begin{align}
        \bigcup\limits_{\ell \in \mathbb{Z}}  [-\frac{1}{16T_s} +  \ell \frac{1}{4T_s}, \frac{1}{16T_s}+\ell \frac{1}{4T_s}  ) \notag
    \end{align}
remains unoccupied.
\end{itemize}

\begin{remark} {\rm
According to \cite{xia2001precoded}, $[x_{i, 0, 0}, x_{i, 1, 0}]^T$ when $M=2, N=8$ and   $[x_{i, 0, 0}, x_{i, 1, 0},x_{i, 2, 0}, x_{i, 3, 0}]^T$ when $M=4, N=4$ are the signal vectors in VOFDM. It aligns with the conclusion derived in the analysis of discrete spectrum of VOFDM  that if a  vector of information symbols is set to $0$,  then a corresponding vector of the same size of the discrete VOFDM signal spectrum is $0$ as well, and  the components of the zero vector are not together but evenly distributed in the spectrum \cite{xia2024discrete}. }
\end{remark}

\subsection{On the Impacts of Interpolation Filter}

The construction of the continuous-time OTFS signal from its samples, i.e., the discrete-time OTFS,  can be expressed as the convolution of the impulse train with the interpolation filter $g_I(t)$. Hence,  \eqref{OTFSbaseband} can be represented as
\begin{align}
{s}(t)  =  \sum_{\eta= -\infty}^{\infty}  {s}_{\eta} \delta_I( t- \eta T_s ) \circledast g_I(t) \notag
\end{align}
where $\circledast$ stands for linear convolution. According to sampling theory \cite{kester2009oversampling, marks2012introduction}, the practical interpolation filter $g_I(t)$ should be a low-pass analog filter\footnotemark, because the output signal ${s}(t)$ must be bandlimited, to prevent imaging. In addition, a wireless communication system is inherently a bandlimited system and must comply with government regulations regarding radio frequencies.

\footnotetext{The pulse of the DAC and the subsequent low pass filter collectively determine the interpolation filter $g_I(t)$. }


To study the impacts of interpolation filter on the PSD of the OTFS signal, we examine its PSD  when $g_I(t)$ is Dirac delta function, sinc function and rectangular function, respectively.
\begin{itemize}
\item \emph{Dirac delta interpolation filter:} When $g_I(t) = \delta(t)$, ${s}_i(t)$ is the discrete-time OTFS, and then the PSD is represented by
\begin{align}
P_c(f) = \sum_{k=0}^{N-1}  \frac{\sigma_{k}^2}{T_s}   \left( \frac{ \sin(\pi(k-fMNT_s) ) }{N \sin(\frac{ \pi}{N} (k-fMNT_s) ) } \right)^2
\end{align}
According to the properties of the squared Dirichlet kernel, the PSD of discrete-time OTFS signals is a periodic function with the period $\frac{1}{MT_s}$,  spanning the entire spectrum.

\item \emph{Sinc interpolation filter:} When $g_I(t) = {\rm sinc}_{\pi}(\frac{t}{T_s})$, $s_i(t)$ is the ideally reconstructed OTFS, which is also known as Whittaker-Shannon interpolation \cite{marks2012introduction}.
     Then, the PSD becomes
\begin{align}
&\quad  {P}_c(f) \notag \\
= & \sum_{k=0}^{N-1}   \frac{\sigma_{k}^2}{T_s}   \left( \frac{ \sin(\pi(k-fMNT_s) ) }{N \sin(\frac{ \pi}{N} (k-fMNT_s) ) } \right)^2   |{\rm rect}(fT_s)|^2 \label{sincRec}
\end{align}
where
\begin{subequations}
\begin{align}
&{\rm sinc}_{\pi}\left(x\right) = \left \{\begin{array}{cc}
                                      1, & {\rm if} \;  x=0, \\
                                      \frac{\sin(\pi x)}{\pi x}, & {\rm if} \; x \neq 0,
                                    \end{array} \notag
    \right.  \\
&{\rm rect}\left(x \right) = \left \{\begin{array}{cc}
                                      0, & {\rm if} \; |x|> \frac{1}{2}, \\
                                      \frac{1}{2}, & {\rm if} \; |x|= \frac{1}{2}, \\
                                      1, & {\rm if} \; |x|< \frac{1}{2}
                                    \end{array}     \right.      \notag
\end{align}
\end{subequations}

The reconstruction filter of sinc function filters the spurious high-frequency ``mirrors" corresponding to  $|f|>\frac{1}{2T_s}$, with an ideal brick-wall-shape low-pass filter. However, sinc function has an infinite response in time domain, which is impossible to perform in practice.

\item \emph{Rectangular interpolation filter:}  When $g_I(t) = {\rm rect}(\frac{t}{T_s})$, $s(t)$ is the sample and hold OTFS, and then the PSD becomes
\begin{align}
&\quad {P}_c(f) \notag \\
= &  \sum_{k=0}^{N-1}   \frac{\sigma_{k}^2}{T_s}   \left( \frac{ \sin(\pi(k-fMNT_s) ) }{N \sin(\frac{ \pi}{N} (k-fMNT_s) ) } \right)^2   |{\rm sinc}_{\pi}(fT_s)|^2
\end{align}
Since the Fourier transform of the rectangular function is the sinc function that has infinite two-sided extent,  the sample and hold OTFS is not bandlimited, and its PSD includes the high frequency images beyond the spectral range $[-\frac{1}{2T_s}, \frac{1}{2T_s})$.  In practice, the rectangular function originates from the rectangular pulse of DAC. Typically, an analog low-pass filter is used to smooth the output of DAC and   suppress the high-frequency images, approximating an ideal brick-wall frequency response.
\end{itemize}
Since this paper focuses on the spectral properties of the OTFS signal, particularly its differences from the OFDM signal, we omit further discussions on other types of interpolation filters, which are universally applicable in digital-to-analog conversions for both OTFS and OFDM signals.

\begin{figure}[tp]{
\begin{center}{\includegraphics[width=6.5cm ]{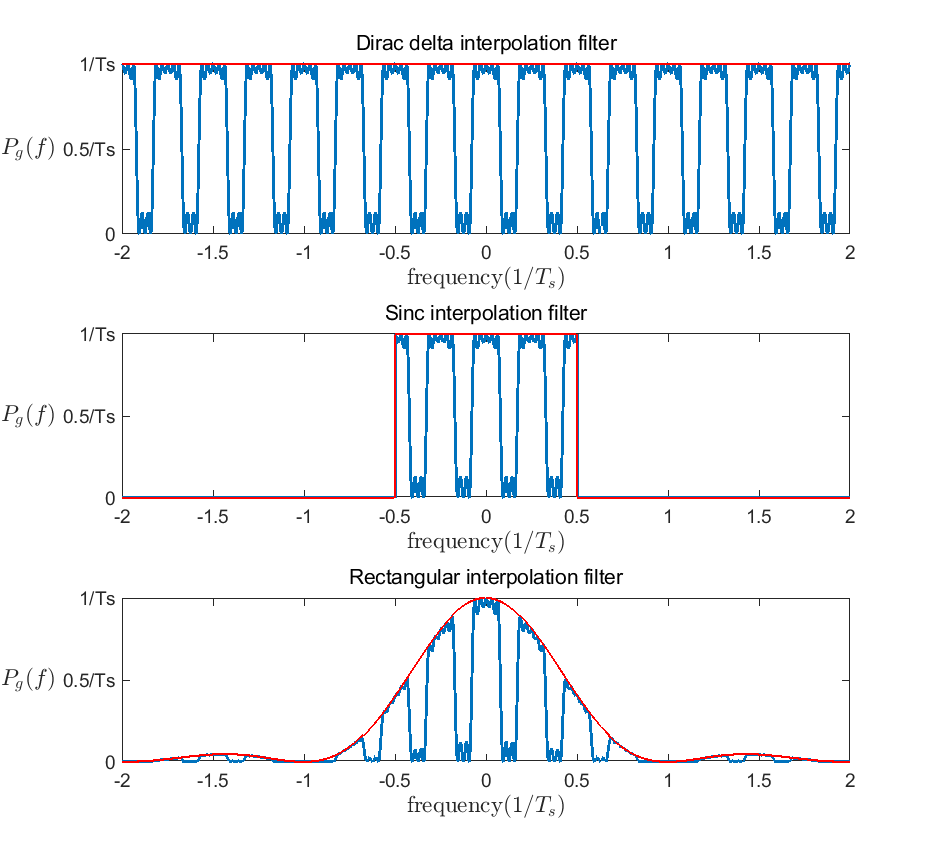}}
\caption{The PSDs of OTFS signals with the different interpolation filters, where the Dirac delta interpolation filter corresponds to discrete OTFS signal, the sinc interpolation filter corresponds to ideal digital-to-analog conversion, and the rectangular filter models the sample-and-hold operation of a DAC.}\label{PSDOTFS}
\end{center}}
\end{figure}

Next, we use an example to illustrate the PSD of OTFS signals.

\textbf{Example 1:} When $N=8$, $M=4$,  $\sigma_k^2 =  1$ for $k \in \{ 0, 1, 2 ,6, 7\}$,  and $\sigma_k^2 = 0$ for  $k \in \{ 3, 4, 5\}$, the PSDs of OTFS signals with the aforementioned interpolation filters are shown in \figref{PSDOTFS}.

\begin{itemize}
\item \emph{Dirac delta interpolation filter:} The spectrum is periodic with the period $\frac{1}{MT_s} = \frac{1}{4T_s}$ and the width of the main lobe is $\frac{2}{MNT_s} = \frac{1}{16T_s}$. Since  $\sigma_k^2 =  1$ for $k \in \{ 0, 1, 2 ,6, 7\}$,  and $\sigma_k^2 = 0$ for  $k \in \{ 3, 4, 5\}$, the occupied spectral range is
\begin{align}
    \bigcup\limits_{\ell \in \mathbb{Z}}  [-\frac{5}{32T_s} +  \ell \frac{1}{4T_s}, \frac{5}{32T_s}+\ell \frac{1}{4T_s}  ) \notag
\end{align}
\item  \emph{Sinc interpolation filter:} The brick-wall-shape low-pass filter ideally filters the spurious high-frequency ``mirrors" corresponding to  $|f|>\frac{1}{2T_s}$, and thus  the occupied spectral range is
\begin{align}
&[-\frac{1}{2T_s}, -\frac{11}{32T_s}) \cup  [-\frac{13}{32T_s}, -\frac{3}{32T_s})   \cup  [-\frac{5}{32T_s}, \frac{5}{32T_s})  \cup \notag \\
& [ \frac{3}{32T_s}, \frac{13}{32T_s}) \cup [\frac{11}{32T_s}, \frac{1}{2T_s}) \notag
\end{align}

\item  \emph{Rectangular interpolation filter:} The interpolation filter amplitude decreases by $3.92$ dB  at $|f| = \frac{1}{2 T_s}$, and the spurious high-frequency ``mirrors"  remain partially present, resulting in distortion in the desired reconstructed waveform.
\end{itemize}

\begin{remark} {\rm
Note that, in practice, a favorable interpolation filter should have a flat passband and a high roll-off to admit the desired frequency components within the desired bandwidth and rejects the unwanted high-frequency ``mirrors".  However, the periodic spectrum of OTFS signals results in the dispersive spread of frequency components over the spectral range $[-\frac{1}{2T_s}, \frac{1}{2T_s})$. This behavior resembles that of a single-carrier system, thereby preventing OTFS waveforms from inheriting OFDM’s key advantage in flexibly managing the signal spectrum.} 
\end{remark}

\section{PSD Comparison of OTFS and OFDM Signals}

In this section, we explore the intrinsic relationships between the PSDs of OTFS and OFDM signals, from which one will see the result in Prop. 2 as well.

For the purpose of comparison, we use the matrix ${\mathbf{X}_i} \in \mathbb{C}^{M \times N}$ to represent the $i$-th information symbol matrix corresponding to $M$ consecutive OFDM symbols, each with $N$ subcarriers, which accords with the $i$-th OTFS symbol, in order to be consistent with signaling of OTFS signals. The same as the setting in Prop. 2, we assume that the information  symbols $x_{i,l,k}$ are i.i.d. with zero mean and variance $\mathbb{E}\left( {x}^*_{i,l,k} {x}_{{i}, {l}, {k}} \right) = \sigma^2_{l,k}$ with respect to the index $i$, and are independent across the indices $(l,k)$, where $l, 0\leq l\leq M-1,$ is the index of $M$ consecutive OFDM symbols in the $i$-th block, and $k$ is the subcarrier index in an OFDM symbol.

\subsection{The Waveform of OFDM Baseband Signal Generated by DAC}

The primary distinction between OFDM and OTFS signals lies in their vectorization manners. Specifically, the vector-form of the discrete baseband signal of OFDM is obtained by row-wise vectorizing the matrix in \eqref{Smatrix}, i.e.,
\begin{align}
 \widehat{\mathbf{s}}_i  = {\rm{vec}}( {\mathbf{S}}_i^T  ) = (\mathbf{I}_M \otimes \mathbf{F}_N^*) {\rm{vec}}( {\mathbf{X}}_i^T)
\end{align}
Feed the discrete-time OFDM signal to a DAC,  the analog signal is obtained as
\begin{align}
& \qquad  \widehat{s}(t) \notag \\
&= \sum_{i= -\infty}^{\infty}  \sum_{n= 0}^{N-1} \sum_{l=0}^{M-1}  {s}_{i, l,n} g_I( t- (iMN+lN+n) T_s ) \notag \\
& = \sum_{i= -\infty}^{\infty}  \sum_{n= 0}^{N-1} \sum_{l=0}^{M-1}  \sum_{k=0}^{N-1}  \frac{{x}_{i, l,k}}{\sqrt{N}}    e^{\jmath 2\pi \frac{k n}{N} }     \cdot   \notag \\
& \qquad \qquad \qquad  g_I( t- (iMN +lN+n) T_s) \label{OFDMbaseband}
\end{align}
The procedures of generating the continuous-time OFDM signal are given in \figref{OFDMwaveform}.


\begin{figure}[tp]{
\begin{center}{\includegraphics[width=7cm ]{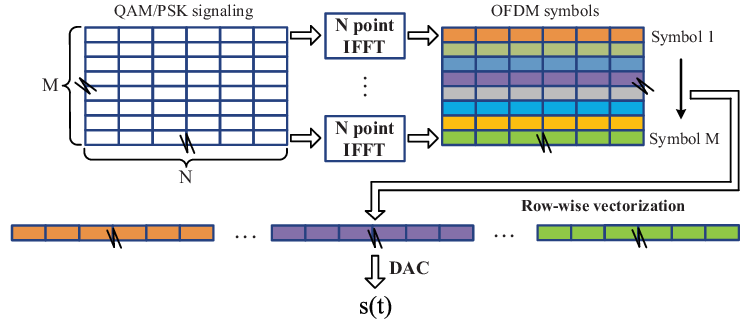}}
\caption{The procedures of generating the continuous-time OFDM signal}\label{OFDMwaveform}
\end{center}}
\end{figure}

\subsection{The PSD of OFDM Signal Generated by DAC}
In the following Prop. 3, we have the PSD expression of the OFDM signal.

\begin{proposition}{\rm
The PSD of the continuous-time OFDM signal $\widehat{s}(t)$ in \eqref{OFDMbaseband} is given by
\begin{align}
 \widehat{P}_c(f) =    \sum_{k=0}^{N-1}  \frac{\widehat{\sigma}_k^2}{T_s}   \left( \frac{ \sin( \pi (k-fNT_s) ) }{ N\sin(\frac{ \pi}{N} (k-fNT_s) ) } \right)^2 |{G}_I(f) |^2
\end{align}
where  $ \widehat{\sigma}^2_{k}  = \sum_{l=0}^{M-1}  \frac{\sigma^2_{l,k}}{M} $  is the signal power of the $k$-th subcarrier. }
\end{proposition}

\begin{proof}
Prop. 3 can be derived using the same method as in Prop. 2. In addition, the PSD of OFDM signal is well-established and the PSD of OFDM signal generated by DAC can be found in \cite{van2010analytical}.
\end{proof}

\subsection{A Study on the Connection Between OTFS and OFDM}

Although the PSD of OTFS signals has been derived directly in Prop. 2, in this subsection, we will see it again from its connection with that of the conventional OFDM signals, which may help understanding OTFS signals better. The relationship between an OTFS signal and its CEP-OFDM  signals in both time domain and frequency domain is illustrated in \figref{OFDMtoOTFS}. The time domain of \figref{OFDMtoOTFS} echoes with \figref{VOFDMwaveform}, and their connection in frequency domain will be investigated in this subsection.

\begin{figure*}[tp]{
\begin{center} 
{\includegraphics[width=14cm]{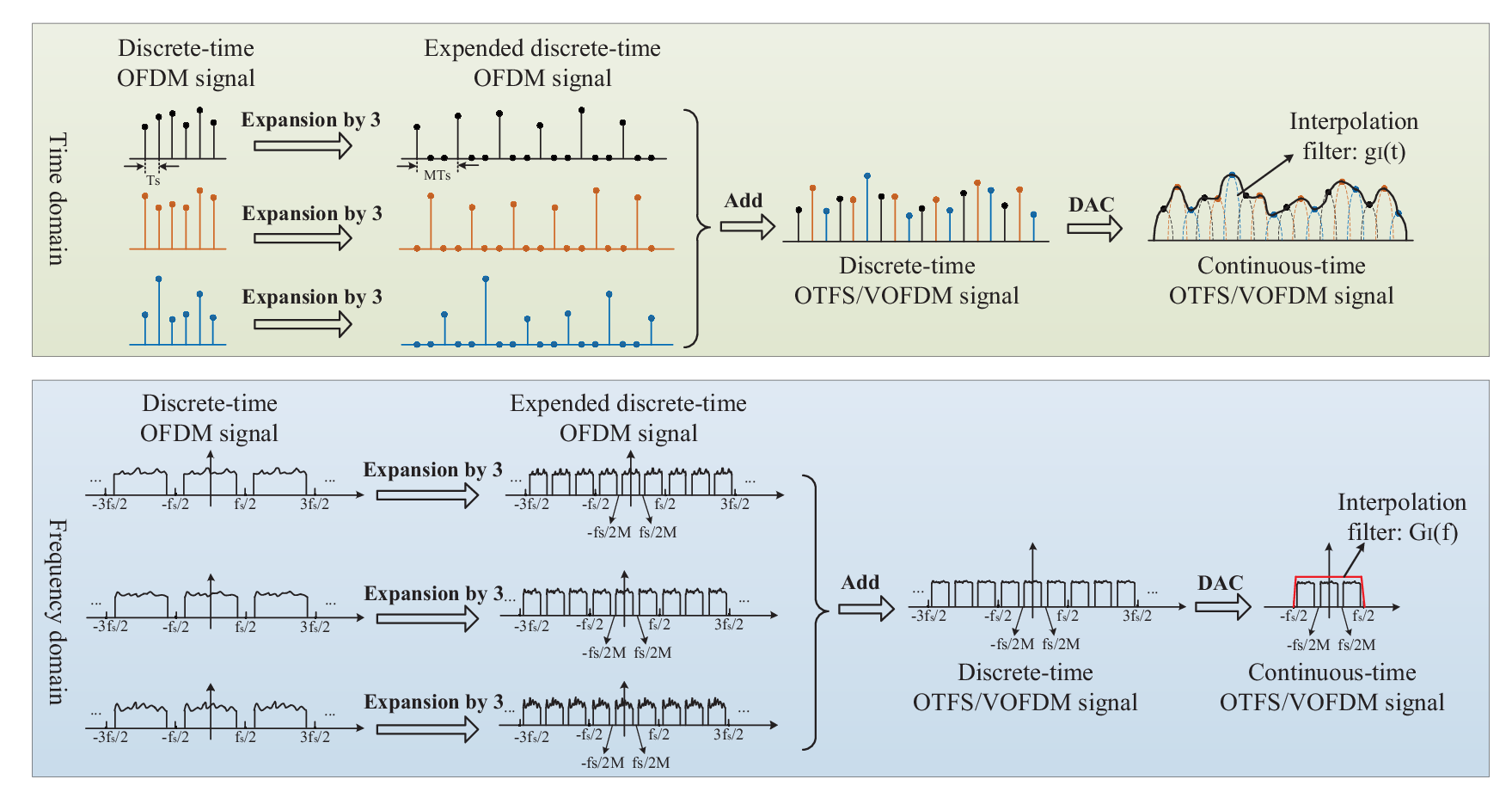}}
\caption{The connection between the OTFS signal and its CEP-OFDM in time and frequency domain}\label{OFDMtoOTFS}
\end{center}}
\end{figure*}

Note that the baseband signal of the OTFS symbol in \eqref{OTFSbaseband} can be represented as a summation of $M$ independent CEP-OFDM signals, i.e.,
\begin{align}
   {s}(t) & =  \sum_{l=0}^{M-1}  \underbrace{ \sum_{i= -\infty}^{\infty} \sum_{n= 0}^{N-1}   {s}_{i,l,n} g_I( t-(iMN+nM+l ) T_s)}_{{\rm  The \; } l {\rm -th \; CEP-OFDM \; signal\;\;}  \widehat{s}^{(l)}(t)} \notag \\
&=  \sum_{l=0}^{M-1}  \underbrace{\sum_{i= -\infty}^{\infty} \sum_{\hat{n}= 0}^{MN-1}   \widehat{s}^{(l)}_{i,\hat{n}} g_I( t - (iMN+ \hat{n}) T_s )}_{ {\rm The \;} l {\rm -th \; CEP-OFDM \; signal\;\;} \widehat{s}^{(l)}(t)} \label{ZSOFDM}
\end{align}
where we term $ \widehat{s}^{(l)}(t)$ as the $l$-th CEP-OFDM signal. For the $i$-th OTFS symbol, its discrete-time sequence $\{  \widehat{s}^{(l)}_{i,\hat{n}}\}|_{\hat{n}=0,\cdots, NM-1} $ is given by
\begin{align}
&\{ \widehat{s}^{(l)}_{i,\hat{n}}\}|_{\hat{n}=0,\cdots, NM-1} \notag \\
=& \Big \{\underbrace{0, \cdots, 0}_{l},  {s}_{i,l,0}, \underbrace{0, \cdots, 0}_{M-l-1}, \cdots, \underbrace{0, \cdots, 0}_{l},  {s}_{i,l,n}, \underbrace{0, \cdots, 0}_{M-l-1},  \notag \\
  &   \cdots, \underbrace{0, \cdots, 0}_{l},  {s}_{i,l,N-1}, \underbrace{0, \cdots, 0}_{M-l-1} \Big\} \label{ZeroStuffed}
\end{align}
which, for fixed $l$ and $i$, is obtained by creating a new sequence $\{ \widehat{s}^{(l)}_{i,\hat{n}}\}|_{\hat{n}=0,\cdots, NM-1}$, comprising the original samples $\{s_{i,l,n}\}|_{n=0,\cdots, N-1}$, separated by $M-1$ zeros and shifted by $l$.  This process aligns with the definition of \emph{expansion} in multirate signal processing \cite{harris2022multirate} in discrete domain. Therefore, we refer to $ \widehat{s}^{(l)}(t)$, which is obtained by passing the expanded discrete-time OFDM sequence ${ \widehat{s}^{(l)}_{i,\hat{n}}}$ through the DAC operating at a sampling rate of $f_s$, as the \emph{CEP-OFDM signal}.  This can also been clearly seen from the VOFDM point of view \cite{xia2000precoded, xia2001precoded}.

Expansion by $M$ in discrete domain indeed reduces the sampling rate from $f_s$ to $\frac{f_s}{M}$, which means the spectrum of the discrete-time CEP-OFDM sequence contracts by the factor $M$. In addition, the interpolation filter $g_I(t)$ remains unchanged. Therefore, the PSD of a CEP-OFDM signal follows immediately from Prop. 3 by substituting  $T_s$ with $MT_s$.

\begin{corollary} {\rm
The PSD of a CEP-OFDM signal $ \widehat{s}^{(l)}(t)$ with $N$ subcarriers and expanded by $M$, as expressed in \eqref{ZSOFDM}, is given by
\begin{align}
 &\widehat{P}_c^{(l)}(f) =  \sum_{k=0}^{N-1}  \frac{{\sigma}_{l,k}^2 }{M T_s} \left( \frac{ \sin( \pi (k-fMNT_s) ) }{ N\sin(\frac{ \pi}{N} (k-fMNT_s) ) } \right)^2 |G_I(f) |^2 \label{ZSOFDMPSD}
\end{align}}
\end{corollary}

Note that, based on \eqref{ZSOFDM}, ${s}(t)$ can be expressed as
\begin{align}
{s}(t) = \sum_{l=0}^{M-1}   \widehat{s}^{(l)}(t)
\end{align}
where, according to \eqref{ZSOFDM} and \eqref{OTFSbaseband}, $\widehat{s}^{({l})}(t) $ is further represented by
\begin{align}
& \widehat{s}^{({l})}(t)  \notag \\
=  & \sum_{i= -\infty}^{\infty}  \sum_{n= 0}^{N-1}  \sum_{k=0}^{N-1}  \frac{{x}_{i, {l},k}}{\sqrt{N}}    e^{\jmath 2\pi \frac{k n}{N} }   g_I( t- (iNM + n M +  {l})T_s)  \notag
\end{align}

When information symbols $ {\mathbf{x}}_{\hat{l}}$ and $ {\mathbf{x}}_{\check{l}}$ in the $\hat{l}$-th and $\check{l}$-th  CEP-OFDM signals, i.e., $\widehat{s}^{(\hat{l})}(t) $ and $\widehat{s}^{(\check{l})}(t) $, are of zero mean and statistically independent for any $\hat{l} \neq \check{l}$,  the PSD of the continuous-time OTFS signal is clearly the summation of the PSDs of $M$ CEP-OFDM signals, i.e.,
\begin{align}
P_c(f) =  \sum_{l=0}^{M-1} \widehat{P}_c^{(l)}(f) \label{sumPSD}
\end{align}
where $\widehat{P}_c^{(l)}(f)$ is the PSD of the $l$-th continuous-time CEP-OFDM signal. 
 
With \eqref{sumPSD} and the assumption on the information symbols $x_{i,l,k}$ that are independent across all indices $i, l, k$,  the PSD of the OTFS signal can also be represented by
\begin{align}
& P_c(f) = \sum_{l=0}^{M-1} \widehat{P}_c^{(l)}(f) \notag \\
    = &  \sum_{k=0}^{N-1} \sum_{l=0}^{M-1} \frac{{\sigma}_{l,k}^2}{M T_s}     \left( \frac{ \sin( \pi (k-fMNT_s) ) }{ N\sin(\frac{ \pi}{N} (k-fMNT_s) ) } \right)^2 |G_I(f) |^2 \notag
\end{align}
which coincides with that obtained in Prop. 2.

\begin{remark} {\rm
A noteworthy point is that the PSD of OTFS signal has the same PSD pattern with the  CEP-OFDM signals when ${\sigma}_{l,k}^2, 0\leq k \leq N-1$, are identical across $l$, indicating that the OTFS signal is essentially a time-interleaved version of $M$ OFDM signals, each with an FFT/IFFT size of $N$.}
\end{remark}

\section{Bandwidth Allocation for OTFS Signals}

In this section, we first review the bandwidth allocation for OFDM signals and then propose a null-space based linear precoding (NSLP) method for OTFS signals to realize flexible bandwidth allocation.

\subsection{A Review of Bandwidth Allocation for OFDM Signals}
For a radio system, the bandwidth of waveforms must comply with the wireless standards. Taking LTE for instance,  its channel bandwidth can be $1.4$MHz, $3$MHz, $5$MHz, $10$MHz, $15$MHz, and $20$MHz\cite{ghosh2010lte}. When the  channel bandwidth is $20$MHz and the bandwidth efficiency is $90\%$, the actually occupied bandwidth is $18$MHz.  To achieve an $18$MHz bandwidth with the subcarrier spacing of $15$kHz, a feasible solution, as suggested in \cite{ShareTechnote}, is to set the following parameters.

\emph{\textbf{Parameters}: The sampling frequency is $f_s = \frac{1}{T_s} = 30.72$MHz, the FFT size is $N_{FFT} = 2048$, with $N_1 = 1201$ occupied subcarriers carrying the QMA/PSK signals, and $N_0 = 847$ guard subcarriers set to zero.}

According to the above parameters, the subcarrier spacing $B_{sc}$ and the occupied bandwidth $B$ are given by
\begin{align}
&B_{sc} = \frac{f_s}{N_{FFT}} = 15kHz, \notag \\
&B = B_{sc} N_1 = 18.015 MHz \notag
\end{align}
From the above example, generating an OFDM waveform that meets the specific requirements of the mobile standard, such as bandwidth and subcarrier spacing, is easily realized by selecting the appropriate sampling frequency and FFT size, and setting a consecutive set of subcarriers to zero.

To summarize, the bandwidth allocation process for OFDM signals is straightforward and can be realized by the \emph{zero-setting method}, as each information symbol corresponds to a single spectral component \cite{christodoulopoulos2011elastic}. In contrast, OTFS signals involve a more complex relationship, where a vector of information symbols corresponds to a vector of spectral components in different and intertwined OFDM symbols. 

To better illustrate the advantages of OFDM over OTFS in  bandwidth allocation, we set $N=32$ to make comparisons with Example 1.

\textbf{Example 2:} When $N=32$,  ${\sigma}_k^2 =  1$ for $0\leq k \leq 9$ and  $22\leq k \leq 31$,  and $\sigma_k^2 = 0$ for  $11 \leq k \leq 21$, the PSDs of OFDM signals with different interpolation filters are shown in \figref{PSDOFDM}.

In Example 2, i.e., OFDM case, $20$ non-zero information symbols correspond to $20$ consecutive spectral components that occupy a bandwidth of $B= N_1 \frac{f_s}{N_{FFT}} = \frac{5}{8T_s} $. By contrast, in Example 1 and its corresponding \figref{PSDOTFS}, i.e., OTFS case, $20$ non-zero information symbols correspond to $20$ spatially separated spectral components that occupy the whole bandwidth of $B = \frac{1}{T_s}$. It indicates that the zero-setting method is ineffective for OTFS signals, and generating an OTFS signal with bandwidth $B$ requires an exact sampling rate of $f_s = B$. In addition, the periodically spread spectrum of OTFS imposes much stricter requirements on the interpolation filter, necessitating a flatter passband and a sharper transition band. Therefore, a flexible bandwidth allocation scheme for OTFS signals is crucial  for its practical applications.

\begin{figure}[tp]{
\begin{center}{\includegraphics[width=6.5cm ]{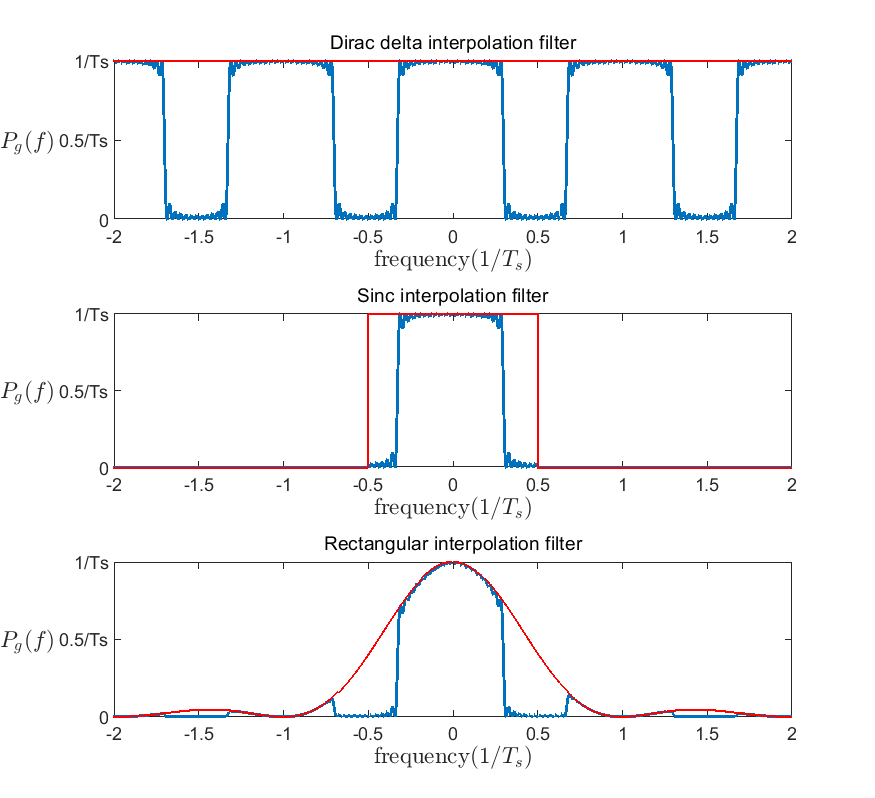}}
\caption{The PSDs of OFDM signals with different interpolation filters, where the Dirac delta interpolation filter corresponds to discrete OFDM signal, the sinc interpolation filter corresponds to ideal digital-to-analog conversion, and the rectangular filter models the sample-and-hold operation of a DAC.}\label{PSDOFDM}
\end{center}}
\end{figure}

\subsection{The Relationship Between Discrete-Time  OTFS Spectrum and the Information Symbols}
To flexibly allocate bandwidth, we need to obtain the relationship between the information symbols $\mathbf{X}_i\in \mathbb{C}^{M \times N}$ and the spectrum of discrete-time OTFS signals. In what follows, since we study a discrete-time OTFS symbol independent of its index, for notational convenience, we drop the OTFS symbol index $i$. 

The spectrum of the discrete-time OTFS symbol $\mathbf{s}$  in \eqref{OTFSdiscrete}  can be derived by applying the $MN$-point FFT, i.e.,
\begin{align}
{\mathbf{y}}  & = {\rm{FFT}_{MN}}({\mathbf{s}} )   = \mathbf{F}_{MN} (\mathbf{F}_N^H \otimes \mathbf{I}_M) {\rm{vec}}( {\mathbf{X}}) \label{FFTotfs}
\end{align}
where $\mathbf{F}_N$ and $\mathbf{F}_{MN}$ are the $N$-point DFT matrix and the $MN$-point DFT matrix, respectively, and
\begin{align}
&{\rm{vec}}(\mathbf{X}) = [x_{0,0}, x_{1,0}, \cdots,x_{M-1,0}, x_{0,1}, x_{1,1},\cdots,\notag \\
&\;\; x_{M-1,1},\cdots,x_{0,N-1}, x_{1,N-1}, \cdots, x_{M-1,N-1}]^T
\end{align}

According to \cite{xia2024discrete} and it is not hard to see by calculating the DFT in \eqref{FFTotfs}, the relationship between the elements of $\mathbf{y}$ and $\mathbf{X}$ is given by
\begin{align}
\mathbf{y}[mN+k] = \frac{1}{\sqrt{M}} \sum_{l=0}^{M-1}(x_{l,k}W_{MN}^{lk})W_M^{lm} \label{DiscreteRelationship}
\end{align}
where $W_{MN}  = e^{ - \frac{j2\pi}{MN}}$, $W_{M}  = e^{ - \frac{j2\pi}{M}}$, and the vector-form of \eqref{DiscreteRelationship} can be written by
\begin{align}
\mathbf{y}^{(k)} =   \frac{1}{\sqrt{M}} \mathbf{F}_M \boldsymbol{\Lambda}^{(k)} \mathbf{x}^{(k)}  \label{DiscreteRelationshipVecForm}
\end{align}
where
\begin{align}
\mathbf{x}^{(k)}  \triangleq   & (x_{0,k}, x_{1,k},  x_{2,k},\cdots,  x_{M-1,k}  )^T, \notag  \\
\mathbf{y}^{(k)} \triangleq  &\big( \mathbf{y}[k], \mathbf{y}[N+k],   \mathbf{y}[2N+k],\cdots,  \notag \\
& \quad \mathbf{y}[(M-1)N+k] \big)^T  \notag
\end{align}
are the information symbol vector and resulting discrete spectrum vector corresponding to the $k$-th subcarrier, respectively, and
\begin{align}
\boldsymbol{\Lambda}^{(k)}  \triangleq & \diag(W_{MN}^{0},W_{MN}^{k}, W_{MN}^{2k},\cdots, W_{MN}^{(M-1)k})\notag
\end{align}

\subsection{Null-space Based Linear Precoding (NSLP) for Flexible Bandwidth Allocation Scheme of OTFS Signals}
The goal of flexible bandwidth allocation is to strategically place information symbols across the $MN$ spectral components according to specific preferences, such as  
\begin{align}
& \qquad \mathbf{y} =  \Big[  \underbrace{\mathbf{y}[0], \cdots,  \mathbf{y}[|\mathcal{J}|/2-1]}_{{\rm assigned \; spectrum:}\; \mathcal{J}}, \underbrace{0,\cdots, 0,}_{{\rm spectral \; nulls:}  \; \mathcal{I},
}   \notag \\
&\qquad  \qquad  \underbrace{\mathbf{y}[  |\mathcal{I}| + |\mathcal{J}|/2 | ], \cdots,  \mathbf{y}[|\mathcal{I}|+|\mathcal{J}|-1]}_{{\rm assigned \; spectrum:}\; \mathcal{J}}  \Big]^T  \label{ZeroSpectrum}
\end{align}
where $\mathcal{I}$ is the index set of the spectral nulls, and $\mathcal{J}$ is the complement of the index set $\mathcal{I}$ that corresponds to the used spectral components.

Mathematically, NSLP is designed to realize the equation \eqref{ZeroSpectrum}. According to our analysis, the spectrum corresponding to a given subcarrier index $k$ exhibits a comb-like structure. A consecutive sequence of spectral components will extend across multiple values of $k$, and thus bandwidth allocation should be performed across multiple subcarrier indices $k$. Furthermore, since Eq. \eqref{DiscreteRelationshipVecForm} indicates that vector $\mathbf{y}^{(k)}$ is only determined by vector $\mathbf{x}^{(k)}$, bandwidth allocation that nulls the spectral components with indices in $\mathcal{I}$ can be performed vector-wisely indexed by $k$ individually and independently.

NSLP is first carried out at a specific value of $k$, and then the procedure is subsequently traversed across all $k$. Let  $\mathcal{I}_k$, where $\mathcal{I}_k  \subset \{0,1,\cdots, M-1 \}$, denote the indices of $\mathcal{I}$ that fall in the index set in the $k$-th vector, and  $\mathcal{J}_k$ denote the complement of the index set $\mathcal{I}_k$ in the set $\{0,1,\cdots, M-1\}$. Mathematically, $\mathcal{I} = \bigcup_{k=0}^{N-1} (\mathcal{I}_k N + k )$ and $ |\mathcal{J}_k| = M -  |\mathcal{I}_k|$. Let $\hat{\mathbf{x}}^{(k)}$ be a $|\mathcal{J}_k| \times 1$ information symbol vector to send, $\mathbf{P}_k$ be an $M \times |\mathcal{J}_k|$ precoding matrix, and ${\mathbf{x}}^{(k)}$ be the precoded $M \times 1$ signal to be transmitted, i.e., 
\begin{align}
{\mathbf{x}}^{(k)} = \mathbf{P}_k \hat{\mathbf{x}}^{(k)} \label{precoder}
\end{align}
such that the following two constraints hold
\begin{subequations}
\begin{align}
& {\rm Constraint \; 1:} \;\;  \left[ \mathbf{y}^{(k)} \right]_{\mathcal{I}_k} = \mathbf{0}_{|\mathcal{I}_k| \times 1}, \label{opt1b} \\
&  {\rm Constraint \; 2:} \;\;  {\rm tr}(\mathbf{P}_k^H \mathbf{P}_k) = |\mathcal{J}_k|  \label{opt1c}
\end{align} \label{opt1}
\end{subequations}
\hspace{-0.29cm} where \eqref{opt1b} is the zero spectrum property as stated in \eqref{ZeroSpectrum}, and  \eqref{opt1c} corresponds to the power constraint. In addition,  $\left[ \mathbf{y}^{(k)} \right]_{\mathcal{I}_k} \in \mathbb{C}^{|\mathcal{I}_k| \times 1}$  denotes the vector composed of the elements of $\mathbf{y}^{(k)} $ with indices in $\mathcal{I}_k$. Substituting \eqref{DiscreteRelationshipVecForm} and \eqref{precoder} into \eqref{opt1b}, it becomes
\begin{align}
{\rm Constraint \; 1:} \;  \left[ \widetilde{\mathbf{F}}^{(k)} \right]_{\mathcal{I}_k}  \mathbf{P}_k \hat{\mathbf{x}}^{(k)} = \mathbf{0}_{|\mathcal{I}_k| \times 1} \label{opt2b} 
\end{align} 
where
\begin{align}
\widetilde{\mathbf{F}}^{(k)} \triangleq \frac{1}{\sqrt{M}} \mathbf{F}_M \boldsymbol{\Lambda}^{(k)} = \left[\widetilde{\mathbf{f}}^{(k)}_0, \widetilde{\mathbf{f}}^{(k)}_1, \cdots, \widetilde{\mathbf{f}}^{(k)}_{M-1}\right]^T  \notag
\end{align}
is an $M\times M$ unitary matrix and $\left[\widetilde{\mathbf{F}}^{(k)} \right]_{\mathcal{I}_k} \in \mathbb{C}^{|\mathcal{I}_k| \times M}$ is the sub-matrix composed of the row vectors of $\widetilde{\mathbf{F}}^{(k)}$ with indices in $\mathcal{I}_k$. 

\textbf{A feasible solution:} From Constraint 1, one can see that $\mathbf{x}^{(k)} = \mathbf{P}_k \hat{\mathbf{x}}^{(k)}$  lies in the null space of $\left[  \widetilde{\mathbf{F}}^{(k)} \right]_{\mathcal{I}_k} $.  Since $\widetilde{\mathbf{F}}^{(k)}$ is an $M\times M$ unitary matrix, a solution of $\mathbf{P}_k$ that meets the above two constraints can be expressed as
\begin{align}
\mathbf{P}_k =  \left[  \widetilde{\mathbf{F}}^{(k)} \right]^H_{\mathcal{J}_k} \mathbf{U}_{|\mathcal{J}_k|} \label{solution1} 
\end{align}
where $ \left[ \widetilde{\mathbf{F}}^{(k)} \right]_{\mathcal{J}_k} \in \mathbb{C}^{|\mathcal{J}_k| \times M}$ is the sub-matrix of the unitary matrix $\widetilde{\mathbf{F}}^{(k)}$ complementary to $\left[\widetilde{\mathbf{F}}^{(k)}\right]_{\mathcal{I}_k}$, and $\mathbf{U}_{|\mathcal{J}_k|} \in  \mathbb{C}^{|\mathcal{J}_k| \times |\mathcal{J}_k| }$ is an arbitrary matrix that satisfies ${\rm tr}(\mathbf{U}_{|\mathcal{J}_k|}^H \mathbf{U}_{|\mathcal{J}_k|}) = |\mathcal{J}_k|$. For the precoder $\mathbf{P}_k$ in \eqref{solution1}, the spectrum null property, or Constraint 1, \eqref{opt1b} or \eqref{opt2b} holds, because
\begin{align}
\left[ \widetilde{\mathbf{F}}^{(k)} \right]_{\mathcal{I}_k}\mathbf{x}_i^{(k)} =
\underbrace{ \left[ \widetilde{\mathbf{F}}^{(k)} \right]_{\mathcal{I}_k} \left[ \widetilde{\mathbf{F}}^{(k)} \right]^H_{\mathcal{J}_k}}_{\mathbf{0}_{|\mathcal{I}_k| \times |\mathcal{J}_k|}}   \mathbf{U}_{|\mathcal{J}_k|} \hat{\mathbf{x}}_i^{(k)}  = \mathbf{0}_{|\mathcal{I}_k| \times 1}  \notag 
\end{align}

\begin{remark} {\rm
The computational complexity  of NSLP mainly arises from precoder design and linear precoding. The main computational cost of precoder design lies in the matrix multiplication in \eqref{solution1}, whose complexity is  ${\mathcal O}( M |\mathcal{J}_k|^2)$. Since \eqref{solution1} is  carried out across all $k$ (where  $1 \leq k \leq  N$), the computational complexity is ${\mathcal O}( MN |\mathcal{J}_k|^2)$. The computational cost of linear precoding is dominated by the matrix multiplication in \eqref{precoder}, and when performed over all $k$, the complexity is ${\mathcal O}(MN|\mathcal{J}_k|)$. Thus, the overall computational complexity is ${\mathcal O}( MN |\mathcal{J}_k|^2)$.} 
\end{remark}

\subsection{Systematic-Form Linear Precoder}

\emph{A special case} is when the precoder $\mathbf{P}_k $ has a systematic form, i.e., 
\begin{align}
\mathbf{P}_k  = \left( \begin{array}{c}
                         \mathbf{I}_{|\mathcal{J}_k|} \\
                         \mathbf{P}_{k,2}
                       \end{array}
 \right)  \label{sysform}
\end{align}
where $\mathbf{I}_{|\mathcal{J}_k|}$ is the $|\mathcal{J}_k| \times |\mathcal{J}_k|$ identity matrix, and $\mathbf{P}_{k,2}$ is an $|\mathcal{I}_k| \times |\mathcal{J}_k|$ matrix.  
By substituting \eqref{sysform} into \eqref{precoder},  the precoded signal ${\mathbf{x}}^{(k)}$ is 
\begin{align}
{\mathbf{x}}^{(k)} =\left( \begin{array}{c}
                         \hat{\mathbf{x}}^{(k)} \\
                         \mathbf{P}_{k,2}\hat{\mathbf{x}}^{(k)}
                       \end{array}
 \right) \label{precoder2}
\end{align}
whose first $|\mathcal{J}_k|$ elements are identical to the information symbol vector $\hat{\mathbf{x}}^{(k)}$. 

To derive the expression of  $\mathbf{P}_k$ to satisfy the spectrum null constraint \eqref{opt2b}, we only need to integrate  \eqref{solution1} and \eqref{sysform}, yielding 
\begin{align}
\left( \begin{array}{c}
                         \mathbf{I}_{|\mathcal{J}_k|}\mathbf{U}_{|\mathcal{J}_k|}^{-1}  \\
                         \mathbf{P}_{k,2} \mathbf{U}_{|\mathcal{J}_k|}^{-1}
                       \end{array}
 \right)  & =  \left[  \widetilde{\mathbf{F}}^{(k)} \right]^H_{\mathcal{J}_k}   = \left( \begin{array}{c}
            \overline{\mathbf{F}}_{k,1} \\
            \overline{\mathbf{F}}_{k,2}  
          \end{array} \right)
   \label{sysform2}
\end{align}
where $\overline{\mathbf{F}}_{k,1}$ is the sub-matrix of $\left[  \widetilde{\mathbf{F}}^{(k)} \right]^H_{\mathcal{J}_k}$ that consists of the first $|\mathcal{J}_k|$ rows of $ \left[  \widetilde{\mathbf{F}}^{(k)} \right]_{\mathcal{J}_k}^H $, and  $\overline{\mathbf{F}}_{k,2}$ is the sub-matrix that consists of the last $|\mathcal{I}_k|$ rows of $ \left[  \widetilde{\mathbf{F}}^{(k)} \right]_{\mathcal{J}_k}^H $. Thus, 
\begin{align}
\left\{ \begin{array}{c}
          \mathbf{U}_{|\mathcal{J}_k|}^{-1} = \overline{\mathbf{F}}_{k,1},  \\
          \mathbf{P}_{k,2} \mathbf{U}_{|\mathcal{J}_k|}^{-1}= \overline{\mathbf{F}}_{k,2}
        \end{array}
\right. \label{PK2}
\end{align} 
where $\mathbf{U}_{|\mathcal{J}_k|}$ should be invertible. According to \eqref{PK2}, $\mathbf{P}_{k,2}$ is derived as $\mathbf{P}_{k,2} = \overline{\mathbf{F}}_{k,2}\overline{\mathbf{F}}^{-1}_{k,1}$. Consequently, the systematic-form $\mathbf{P}_k$ is expressed as
\begin{align}
\mathbf{P}_k  = \left( \begin{array}{c}
                         \mathbf{I}_{|\mathcal{J}_k|} \\
                          \overline{\mathbf{F}}_{k,2}\overline{\mathbf{F}}^{-1}_{k,1}
                       \end{array}
 \right)  \label{sysform3}
\end{align} 
The precoder $\mathbf{P}_k $ in \eqref{sysform3} should then be normalized to meet the power constraint of \eqref{opt1c}.  This systematic precoder is the same as that in \cite{xia2024discrete}. Comparing with the precoder in \eqref{solution1}, the precoder in \eqref{sysform3} requires a matrix inverse that may change the signal peak-to-average power ratio (PAPR) for the transmission as studied in \cite{xia2024discrete}. In contrast, the precoder in \eqref{solution1} only does partial unitary matrix multiplications and does not require any matrix inverse, and thus may have a better PAPR property than that for the precoder in  \eqref{sysform3}.

\begin{figure*}[htp]
\begin{minipage}[!h]{.33\linewidth}
\centering
\includegraphics[width=5.5cm]{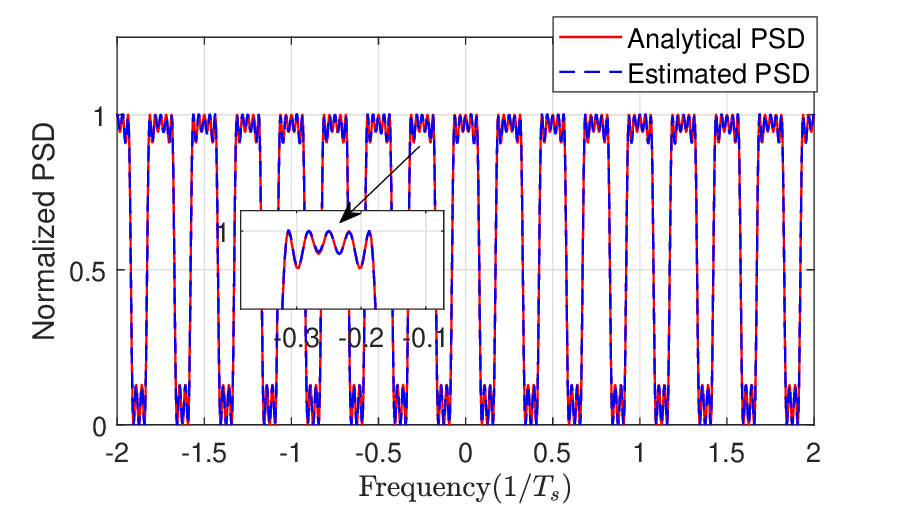}
\subcaption{ Dirac delta interpolation filter }
\label{xx1}
\end{minipage}
\begin{minipage}[!h]{.33\linewidth}
\centering
\includegraphics[width=5.5cm]{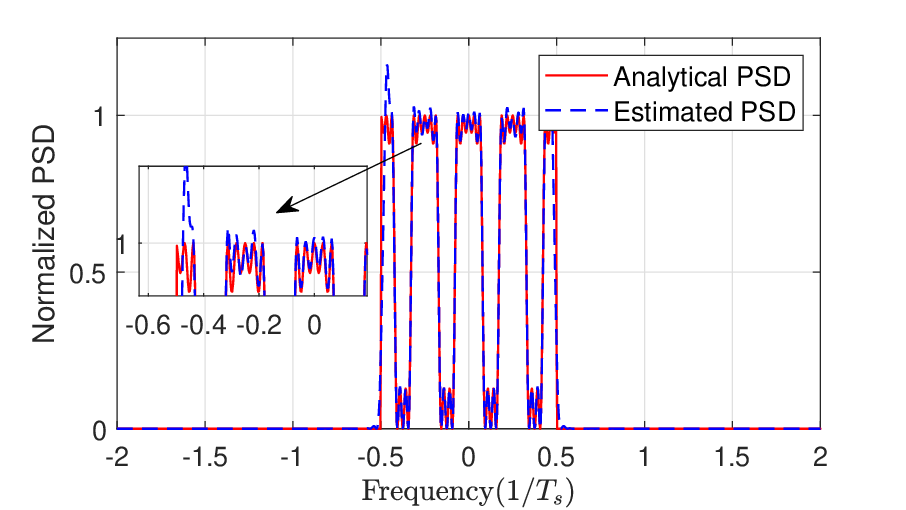}
\subcaption{Truncated sinc interpolation filter}
\label{xx2}
\end{minipage}
\begin{minipage}[!h]{.33\linewidth}
\centering
\includegraphics[width=5.5cm]{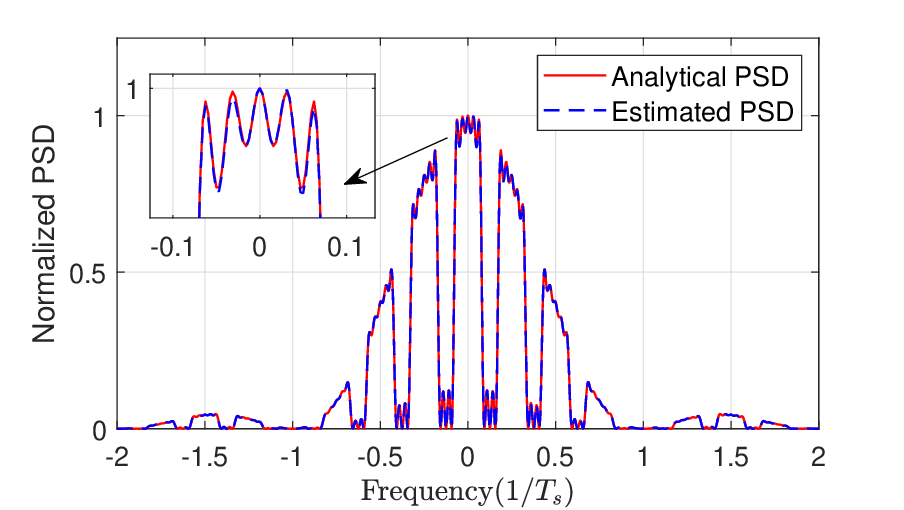}
\subcaption{Rectangular interpolation filter}
\label{xx3}
\end{minipage}
\caption{Comparisons of our derived analytical PSD with the estimated PSD of OTFS signals}
\label{PSDval}
\end{figure*}

\begin{remark} {\rm
In a more general scenario, where the dimension of the information symbol vector $\hat{\mathbf{x}}_i^{(k)}$, denoted as $L$,  is smaller than $|\mathcal{J}_k| $, i.e., $ L < |\mathcal{J}_k| $, the solution in \eqref{solution1} still holds by adjusting the size of $\mathbf{U}_{|\mathcal{J}_k|}$ from ${|\mathcal{J}_k| \times |\mathcal{J}_k|  }$  to  ${|\mathcal{J}_k| \times L }$. With a more general setting, more degree of freedoms may exist for the design, which may benefit for a particular application in practice.}
\end{remark}

\begin{remark} {\rm
The above precoding only depends on a given spectrum null requirement and a given information symbol vector  $\hat{\mathbf{x}}_k$ to send. However, when the channel state information is available, the precoder design can optimize power allocation across multiple spectral components, such as through water-filling power allocation, similar to the OFDM case, to enhance performance over a fading channel.}
\end{remark}

\begin{remark} {\rm
With the systematic-form linear precoder, the receiver can extract the information-bearing symbols $\mathbf{x}^{(k)}$ directly from the received codeword, especially when the channel is high-quality. This is critical for latency-sensitive applications (e.g., real-time video streaming) as it is computationally efficient.
}
\end{remark}

\section{Numerical Results}

In this section, we validate our derived PSD expression of OTFS signals and evaluate the capability of our proposed bandwidth allocation scheme. The amplitude of the estimated PSD is normalized in order to allow a qualitative comparison \cite{talbot2008spectral}.

\subsection{Validation of the Analytical Expressions for OTFS PSD}
In \figref{PSDval} and Table. \ref{T1}, we compare our derived analytical PSD expression with the estimated PSD. The analog OTFS signal is obtained with an oversampling rate of $100f_s$, effectively simulating $T_s \rightarrow 0$. In \figref{PSDval}, we plot the analytical PSD curve and the estimated PSD curve using the parameters in Example 2. We adopt the periodogram method for PSD estimation, which first divides the signal into non-overlapping segments, and then performs DFT for each segment, and finally averages them. The sampling frequency of the periodogram method is $100f_s$. Figs. \ref{PSDval}(a)-(c) correspond to the three interpolation filters described in Section II. C. Note that, as the ideal sinc function extends across the entire time domain, we use the truncated sinc function with the order of $50$ in our simulation. From \figref{PSDval}, we can see that the three PSD curves corresponding to the analytical expression in Prop. 2 closely follow the estimated PSD curve for the simulated signal. The difference of OTFS PSD with the sinc interpolation filter is more significant as the actual interpolation filter in the simulated signal is a truncated sinc function, instead of a full sinc function. Besides, we use normalized mean squared error (NMSE) and cosine similarity to quantitatively measure the differences between the analytical expressions and the estimated PSD curves in Table. \ref{T1}.

\begin{table}[tp] \centering
\noindent
\caption{Differences between the analytical PSD and the estimated PSD of OTFS signals with different interpolation filters} \label{T1}
 \begin{tabular} {|p{2.75cm}<{\centering}|p{1.3cm}<{\centering}|p{1.75cm}<{\centering}|p{1.25cm}<{\centering}|}
\hline\hline
\backslashbox{Metric}{Interpolation}    & Dirac delta   &  Truncated sinc   & Rectangular  \\ \hline
NMSE(dB)  &     -48.9872 &    -18.0664     &    -47.6115\\\hline
Cosine Similarity & 0.99999369    &  0.99221525    & 0.99999440  \\\hline
\end{tabular}
\end{table}

\subsection{Connection Between OTFS Signals and CEP-OFDM Signals}

In this subsection, we validate the relationship between the PSD of the OTFS signal and the PSD of the CEP-OFDM signals derived in \eqref{sumPSD}. We use the truncated sinc interpolation function to generate the analog OTFS signal. To be persuasive, we use the estimated PSD for comparison.

\begin{figure}[tp]{
\begin{minipage}[!h]{1\linewidth}
\centering
\includegraphics[width=6cm]{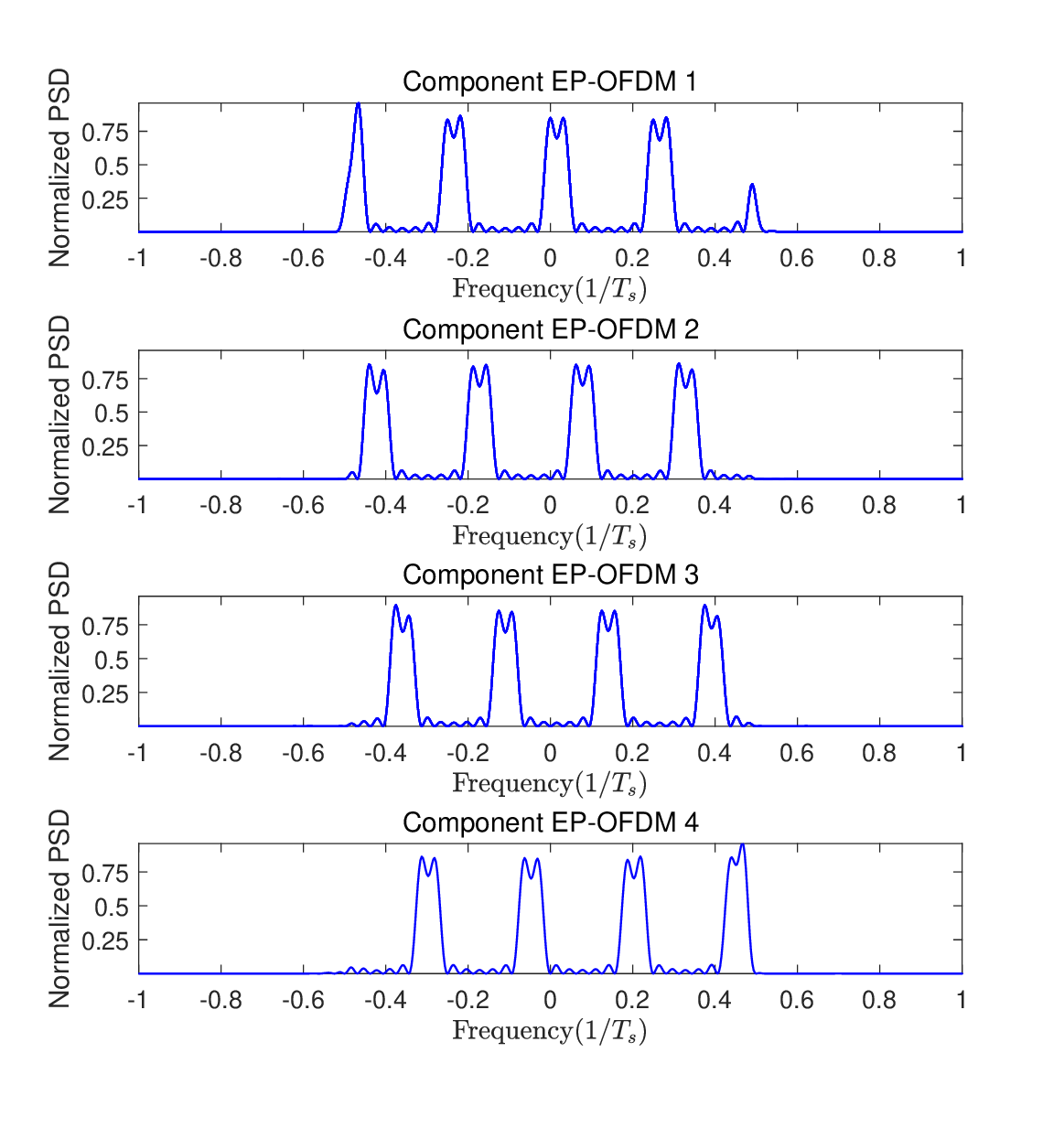}
\subcaption{The estimated PSD of the CEP-OFDM signals}
\label{ComponentZSOFDMv2}
\end{minipage}
\begin{minipage}[!h]{1\linewidth}
\centering
\includegraphics[width=6cm]{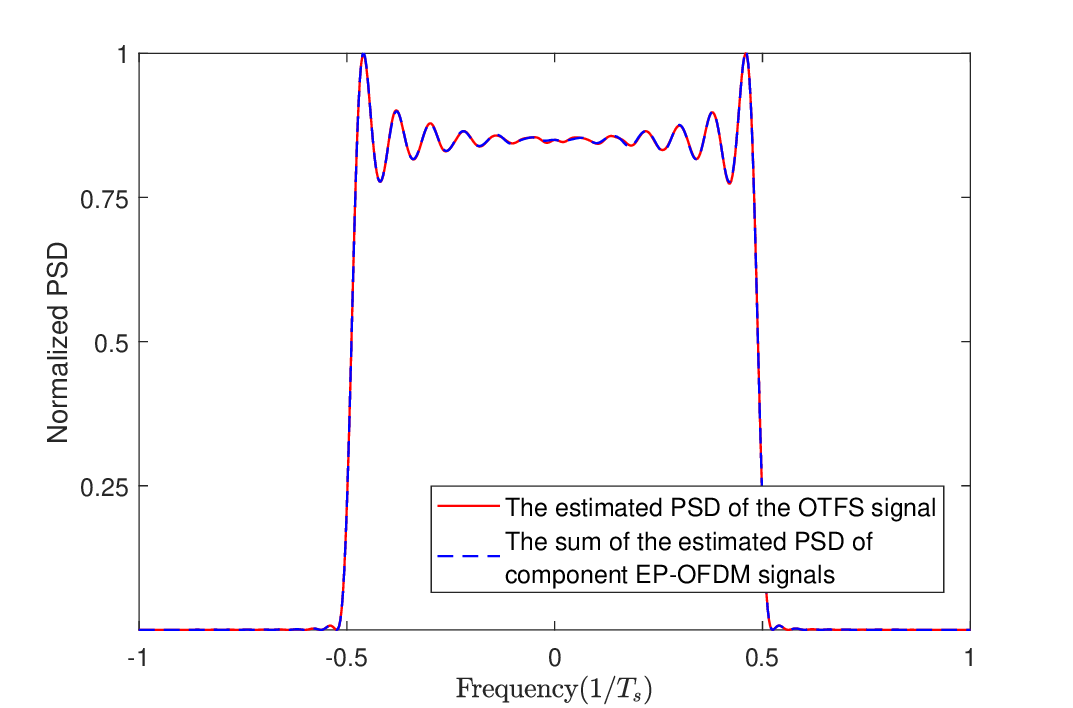}
\subcaption{The estimated PSD of the OTFS signal and the sum of the estimated PSD of the CEP-OFDM signals}
\label{ComponentZSOFDMv2}
\end{minipage}}
\caption{The estimated PSD curves of the OTFS signal and the CEP-OFDM signals when the information matrix is of the pattern $\mathbf{X}_1$ in \eqref{X2pattern}} \label{PSDotfsofdmv2}
\end{figure}

In \figref{PSDotfsofdmv2}, the parameters are  $N=8$ and $M=4$, and the information matrix is of the pattern
\begin{align}
&\qquad \mathbf{X}_1 = \notag \\
 &\left(
                 \begin{array}{cccccccc}
                   X_{0,0} & X_{0,1} & 0 & 0 & 0 & 0 & 0 & 0 \\
                   0 & 0 & X_{1,2} & X_{1,3} & 0 & 0 & 0  & 0  \\
                   0 & 0 & 0 & 0 & X_{2,4} & X_{2,5} & 0  & 0\\
                   0 & 0 & 0 & 0 & 0 &  0 &  X_{3,6}  & X_{3,7} \\
                 \end{array}
               \right) \label{X2pattern}
\end{align}
where $X_{m,n} $ is independently drawn from the QPSK constellation with unit power. Thus, $\sigma_k^2 = \frac{1}{4}$ for  $k \in \{1, 2, 3, 4, 5, 6, 7, 8\}$. The OTFS signal used for PSD estimation is composed of $10^6$ consecutive OTFS symbols. From \figref{PSDotfsofdmv2}(a), we can see that the PSD curves of the $M=4$ CEP-OFDM signals are non-overlapping, which is a result of the block diagonal structure of $\mathbf{X}_1$. From \figref{PSDotfsofdmv2}(b), we can also find that the PSD of the OTFS signal is almost identical to the sum of the estimated PSD of the CEP-OFDM signals, which verifies the conclusion in \eqref{sumPSD}. It is noteworthy that the PSD of the OTFS signal corresponding to $\mathbf{X}_1$ is exactly the squared Fourier transform of the truncated sinc interpolation filter. This is because $\sigma_k^2$ remains constant for any $k$.

Note that the equivalence between the sum PSD of the CEP-OFDM signals and the PSD of the OTFS signal relies on the cross correlation of the information sequences corresponding to different CEP-OFDM signals. Since the estimation of PSD is based on the input time-domain signal, the signal's length is crucial for determining  the cross correlation. Thus, in \figref{NMSEcos}, we investigate the discrepancy between the PSD of the OTFS signal and  the sum of the estimated PSD of the CEP-OFDM signals. It can be seen that the NMSE decreases as the signal length increases, while the cosine similarity increases as the signal length increases. This implies that when the signal length is infinite, the equivalence holds almost surely.
\begin{figure}[htp]{
\begin{center}{\includegraphics[width=6cm ]{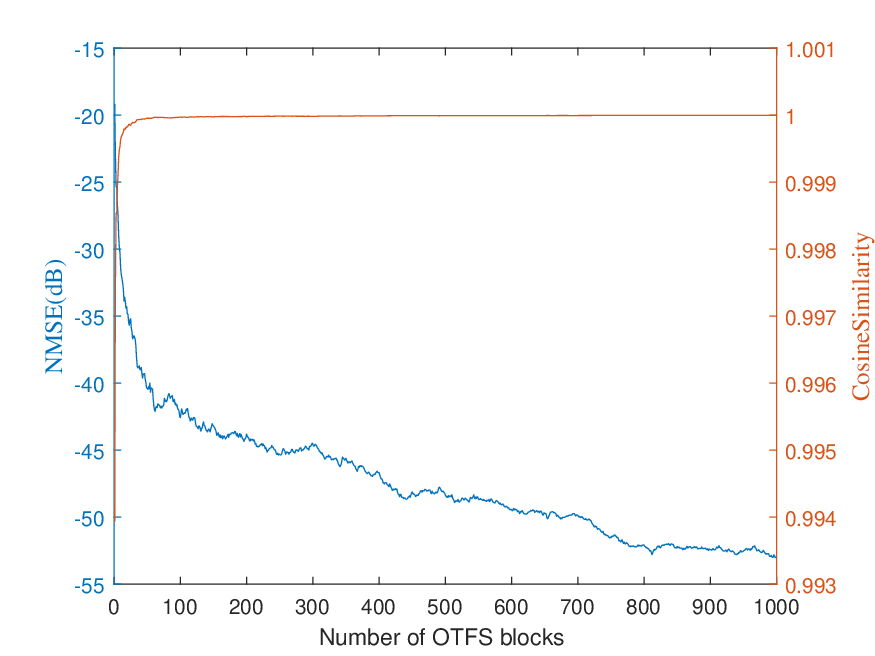}}
\caption{NMSE/Cosine similarity versus OTFS block number (each OTFS block consists of $100$ OTFS symbols)}\label{NMSEcos}
\end{center}}
\end{figure}

\subsection{Bandwidth Allocation of OTFS Signals}

One of the most frequently cited advantages of OTFS is its architectural compatibility with various multicarrier modulation schemes, including traditional OFDM\cite{hadani2018otfs, OTFS}. However, as analyzed in this paper, merely substituting the baseband sequence of the OFDM signal with that of the OTFS signal does not adequately meet the spectral requirements of OFDM-related wireless standards. To numerically validate the effectiveness of our proposed bandwidth allocation scheme for OTFS and its compatibility with the specifications of OFDM-based systems,  we adopt the OFDM parameters from Sec. IV(A), i.e.,  $f_s = \frac{1}{T_s} = 30.72$MHz, the required bandwidth of $18$MHz, and  the subcarrier spacing of $15$kHz, to illustrate the spectral properties of OTFS signal with different zero-setting patterns of the information matrix $\mathbf{X} \in \mathbb{C}^{M\times N}$, where $M=16, N=128$. To facilitate comparison, we set the number of non-zero entries in $\mathbf{X}$ to $1201$, the same as the OFDM case discussed in Section IV(A).
\begin{figure}[tp]{
\begin{minipage}[!h]{1\linewidth}
\centering
\includegraphics[width=6cm]{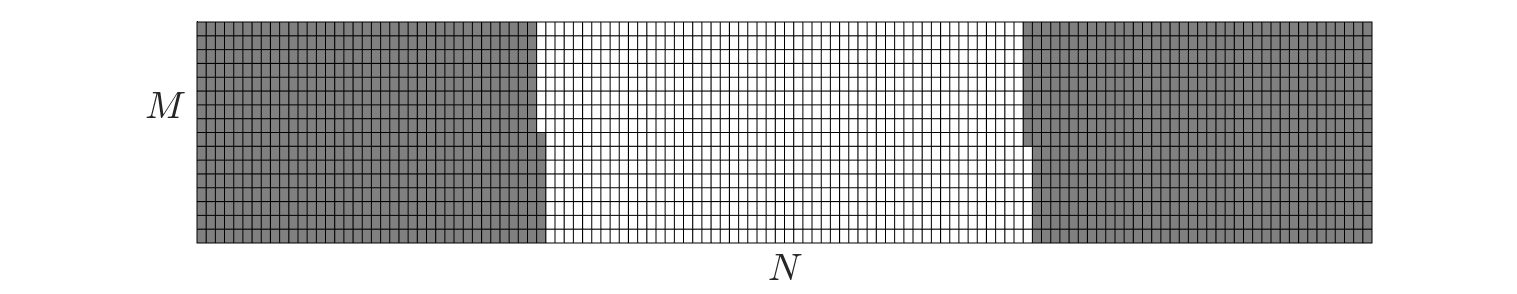}
\subcaption{Zero-setting pattern $1$ of the information matrix}
\label{QAMpatternV1a}
\end{minipage}
\begin{minipage}[!h]{1\linewidth}
\centering
\includegraphics[width=6cm]{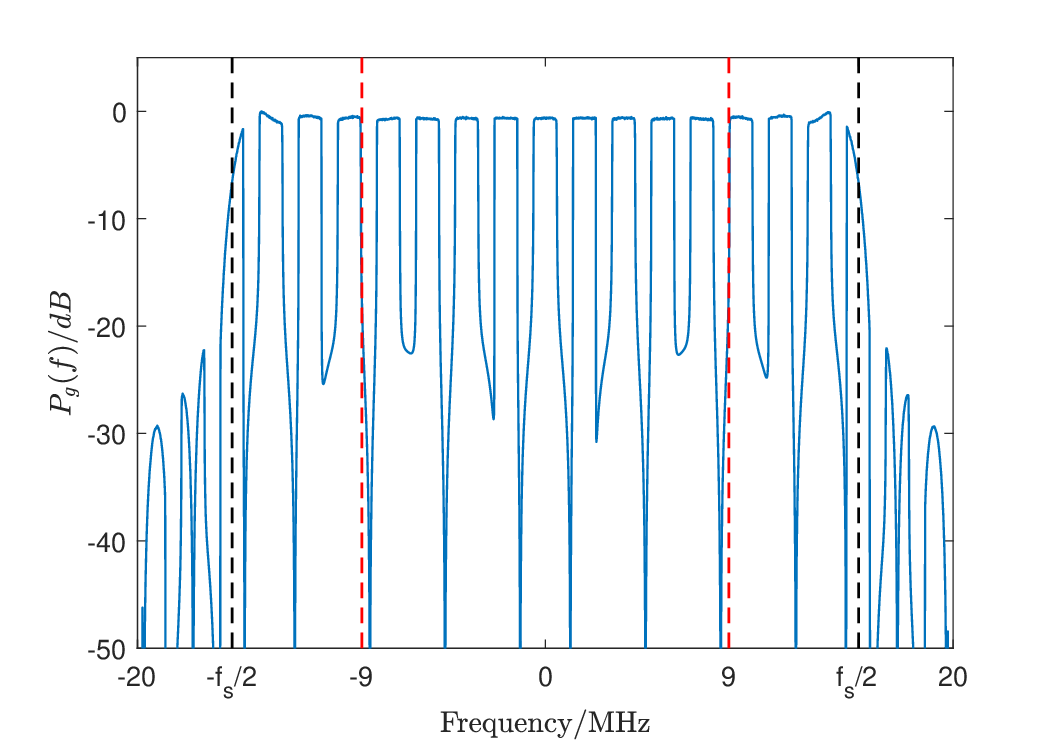}
\subcaption{The corresponding PSD}
\label{QAMpatternV1b}
\end{minipage}}
\caption{Zero-setting pattern $1$ of the information matrix and the corresponding PSD (The blue dotted line represents the spectrum range of the DAC with an ideal sinc interpolation filter, while the red dotted line represents the desired spectrum range.)} \label{QAMpatternV1}
\end{figure}

In \figref{QAMpatternV1},  $1201$ entries located in the head and tail columns of $\mathbf{X}$ are used to carry the information symbols, and the rest $847$ entries located in the middle columns are mapped to zero, which is shown in \figref{QAMpatternV1}(a). Recall that the PSD of the OTFS signal is the sum PSD of the CEP-OFDM signals. Since the symbol-to-subcarrier zero-setting of $M$ CEP-OFDM signals, i.e., columns in \figref{QAMpatternV1}(a), are nearly identical in  \figref{QAMpatternV1}(a), the PSD of the OTFS signal retains a similar shape to that of a CEP-OFDM signal. Because the middle subcarriers of the CEP-OFDM signal are set to zero, the PSD of the corresponding OTFS signal takes on the comb shape depicted in \figref{QAMpatternV1}(b). Another noteworthy point is that the spectrum of the resulting OTFS signal occupies the entire spectral range, i.e., $[-f_s/2, f_s/2]$.

In \figref{QAMpatternV2}, $1201$ entries located in the head and tail rows of $\mathbf{X}$  are used to carry the  information symbols, and the rest $847$ entries located in the middle rows are mapped to zero, which is shown in \figref{QAMpatternV2}(a). Unlike Pattern 1 in \figref{QAMpatternV1}(a), the CEP-OFDM symbols in Pattern 2 in \figref{QAMpatternV2}(a) are categorized into three types: non-zero, partial-zero, and full-zero cases. Since the PSD of the OTFS signal is the sum PSD of the CEP-OFDM signals, the resulting PSD is nearly flat across the entire spectral range of $[-f_s/2, f_s/2]$, as shown in \figref{QAMpatternV2}(b).

\begin{figure}[tp]{
\begin{minipage}[!h]{1\linewidth}
\centering
\includegraphics[width=6cm]{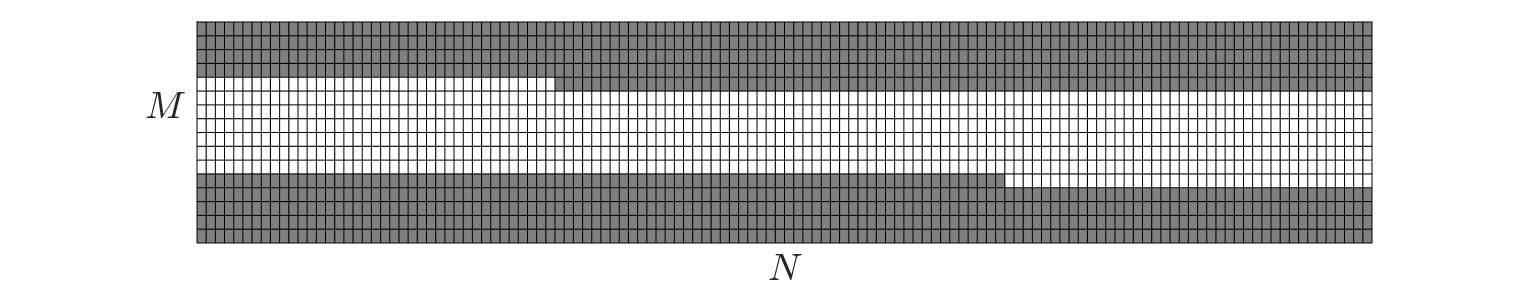}
\subcaption{Zero-setting pattern $2$ of the information matrix}
\label{QAMpatternV2a}
\end{minipage}
\begin{minipage}[!h]{1\linewidth}
\centering
\includegraphics[width=6cm]{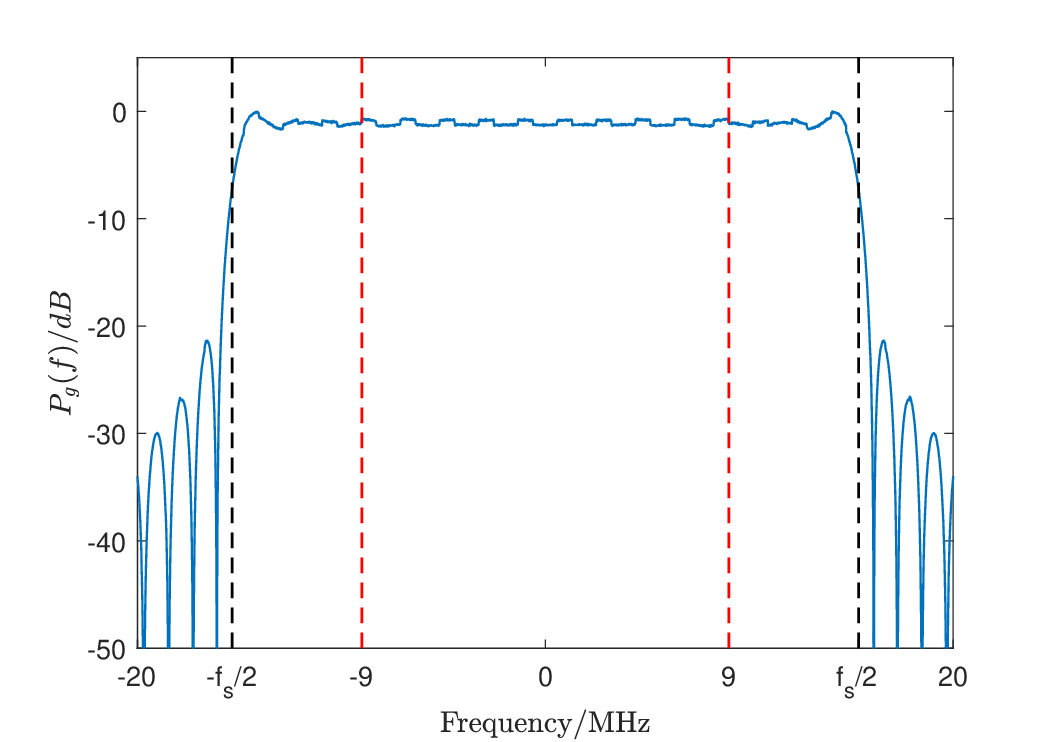}
\subcaption{The corresponding PSD}
\label{QAMpatternV2b}
\end{minipage}}
\caption{Zero-setting pattern $2$ of the information matrix and the corresponding PSD  (The blue dotted line represents the spectrum range of the DAC with an ideal sinc interpolation filter, while the red dotted line represents the desired spectrum range.)} \label{QAMpatternV2}
\end{figure}

The results in \figref{QAMpatternV1} and \figref{QAMpatternV2} indicate that simply setting zeros to the entries of the information matrix $\mathbf{X}$, as in OFDM bandwidth allocation, does not effectively customize the spectrum, leading to a waveform that occupies the entire spectral range. Therefore, it is essential to implement a flexible bandwidth allocation scheme for OTFS signals to ensure compatibility with current wireless standards, which are mostly based on multi-carrier modulation. In \figref{NSLP},  we validate the effectiveness of our proposed NSLP based  bandwidth allocation scheme for OTFS signals.  In particular, we set the $\mathbf{U}_{|\mathcal{J}_k|}$ in the precoder of \eqref{solution1} as an identity matrix of appropriate size. In order to achieve the desired bandwidth, an appropriate zero-setting, i.e., the proper selection of $\mathcal{I}_k, k=0,\cdots, N-1$, is also essential. By first determining the indices of the occupied discrete spectral components of $\mathbf{y}$ according to \eqref{ZeroSpectrum} and then transforming  the subscript $(m, k)$ of $\mathbf{y}$ to the subscript $(l,k)$ of the information-bearing matrix $\mathbf{X}$  according to \eqref{DiscreteRelationship}, the appropriate zero-setting that corresponds to the desired spectrum range of $[-9\text{MHz}, 9\text{MHz})$ is found to be Pattern $2$ of $\mathbf{X}$ in \figref{QAMpatternV2}(a). Then, based on Pattern $2$, we perform NSLP design. Taking the $1$-st subcarrier in \figref{QAMpatternV2}(a) as an example, the entries with indices in $\mathcal{J}_0 = \{1,2,3,4, 12,13,14,15,16\}$ are non-zero, with each entry corresponding to the $M$ spikes at the same positions in the frequency domain. Through precoding, the NSLP method effectively suppresses the $|\mathcal{I}_0|=7$ spectral components in the unwanted region, i.e., $[-f_s/2, -9\text{MHz}) \cup [9\text{MHz}, f_s/2)$, while maintaining the $|\mathcal{J}_0|=9$ spectral components in the desired region, i.e., $[-9\text{MHz}, 9\text{MHz})$. Then, the precoding procedure is applied across the rest $N-1$ subcarriers of the CEP-OFDM signals, resulting in a continuous and flat spectrum within the range of $[-9\text{MHz}, 9\text{MHz})$.

\begin{figure}[tp]{
\begin{center}{\includegraphics[width=6cm ]{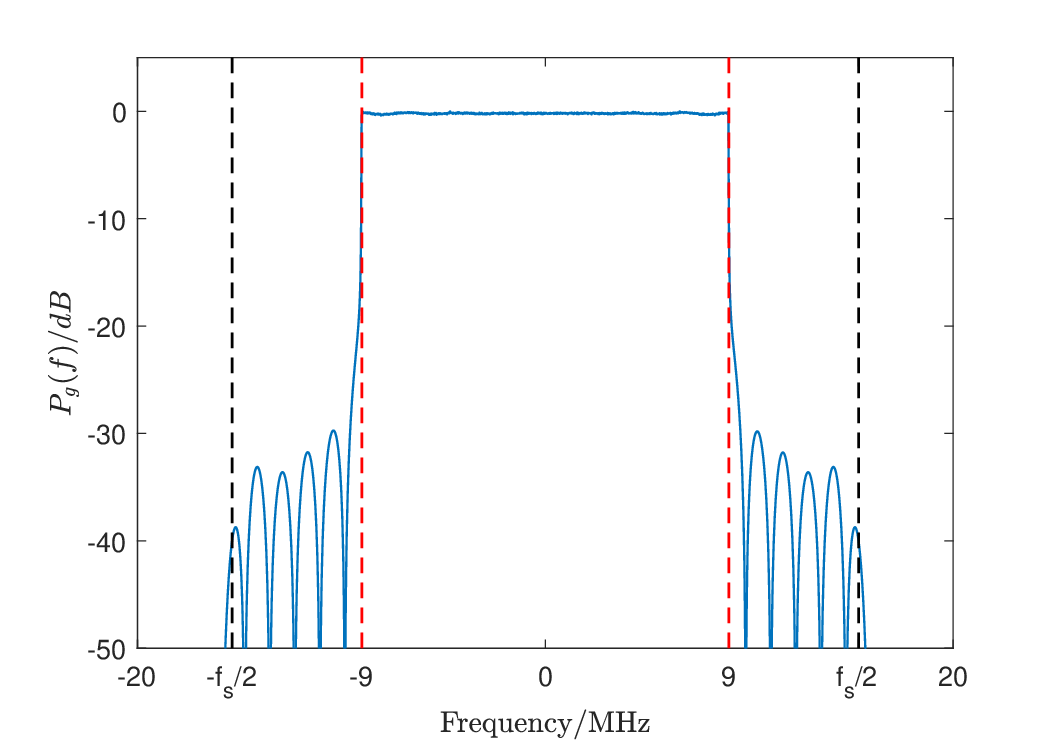}}
\caption{The PSD of the OTFS signal  tailored by our proposed NSLP method  (The blue dotted line represents the spectrum range of the DAC with an ideal sinc interpolation filter, while the red dotted line represents the desired spectrum range.)}\label{NSLP}
\end{center}}
\end{figure}

\section{Conclusions}
In this paper, we carry out PSD analysis of OTFS signals. We demonstrate that the PSD of discrete OTFS signals is periodic with a period of $\frac{1}{MT_s}$, resulting in $M$ identical spectral components within the range $[-\frac{1}{2T_s}, \frac{1}{2T_s})$. We further express that, when the information symbols are independent, the PSD of OTFS signals is equal to the sum of the PSDs of the CEP-OFDM signals. In addition,  we propose a null-space-based linear precoding method for OTFS signals to enable flexible bandwidth allocation. Numerical results validate our analytical results regarding the PSD of OTFS signals and show the effectiveness of our proposed NSLP method in tailoring the spectrum of OTFS signals.

\begin{appendices}

\section{Proof of Proposition 1}
The autocorrelation function of $s_{\eta}$ is given by
\begin{align}
& R_{ss}(\eta, \hat{\eta}) = \mathbb{E}( {s}^*_{\eta}{s}_{\hat{\eta}}  )  \notag \\
 = &  \sum_{i = -\infty}^{\infty} \sum_{\hat{i} = -\infty}^{\infty}  \sum_{n = 0}^{N-1} \sum_{\hat{n} = 0}^{N-1} \sum_{l=0}^{M-1} \sum_{\hat{l}=0}^{M-1}   \sum_{k=0}^{N-1}  \sum_{\hat{k}=0}^{N-1}  \frac{ \mathbb{E} ( {x}^*_{i, l,k}{x}_{\hat{i}, \hat{l}, \hat{k}}  )  }{N}     \notag \\
   &    e^{\jmath 2\pi \frac{\hat{k} \hat{n}}{N} }  e^{-\jmath 2\pi \frac{k n}{N} }    \delta [\eta - iMN - nM - l]    \notag \\
   & \delta [\hat{\eta} - \hat{i}MN - \hat{n}M - \hat{l}]  \notag \\
 = &  \sum_{i = -\infty}^{\infty}  \sum_{n = 0}^{N-1} \sum_{\hat{n} = 0}^{N-1} \sum_{l=0}^{M-1}  \sum_{k=0}^{N-1}  \frac{ \sigma^2_{l,k} }{N}   e^{-\jmath 2\pi \frac{k (n-\hat{n})}{N} }    \notag \\
  &      \delta [\eta - iMN - nM - l]    \delta [\hat{\eta} -  {i}MN - \hat{n}M -  {l}]
\end{align}
In addition, the mean of  $s_{\eta}$ is given by  $\mathbb{E}(s_{\eta}) \equiv 0$.

Since $R_{ss}(\eta,  \hat{\eta})  = R_{ss}(\eta + MN,  \hat{\eta}  + MN) $ and $\mathbb{E}(s_{\eta})  = \mathbb{E}(s_{\eta+ \Delta \eta})   $, $\{s_{\eta}\}$ is \emph{cyclostationary} with a period of $MN$.

\section{Proof of Proposition 2}

The autocorrelation function of $ {s}(t)$ is
\begin{align}
&\quad \phi(t, \tau) =  \mathbb{E}( {s}^*(t)  {s}(t+\tau) ) \notag \\
& =   \sum_{i = -\infty}^{\infty} \sum_{\hat{i} = -\infty}^{\infty}  \sum_{n = 0}^{N-1} \sum_{\hat{n} = 0}^{N-1} \sum_{l=0}^{M-1} \sum_{\hat{l}=0}^{M-1}   \sum_{k=0}^{N-1}  \sum_{\hat{k}=0}^{N-1}  \frac{ \mathbb{E} ( {x}^*_{i, l,k}{x}_{\hat{i}, \hat{l}, \hat{k}}  )  }{N}     \notag \\
   & \quad   e^{\jmath 2\pi \frac{\hat{k} \hat{n}}{N} }  e^{-\jmath 2\pi \frac{k n}{N} }    g_{I}( t - (i NM + nM + l)  T_s)    \notag \\
   &\quad  g_{I}( t - (i NM + \hat{n}M + l)T_s  + \tau)   \notag \\
& =   \sum_{i= -\infty}^{\infty}   \sum_{n= 0}^{N-1}  \sum_{\hat{n}= 0}^{N-1}    \sum_{l=0}^{M-1}   \sum_{k=0}^{N-1}   \frac{ \sigma_{l,k}^2}{N}  e^{-\jmath 2\pi \frac{k (n - \hat{n})}{N} }  \notag \\
&     g_{I}( t - (i NM + nM + l)  T_s) g_{I}( t - (i NM + \hat{n}M + l)T_s  + \tau)
\label{OTFSbasebandxx}
\end{align}
where \eqref{OTFSbasebandxx} is obtained as
\begin{align}
\mathbb{E}({x}^*_{i,l,k}{x}_{\hat{i},\hat{l},\hat{k}})  = \left\{   \begin{array}{cc}
                                                                                        \sigma_{l,k}^2, & l = \hat{l}\; {\rm and} \; k = \hat{k} \\
                                                                                        0, & {\rm otherwise }
                                                                                      \end{array} \notag
\right.
\end{align}
In addition, the mean of $s(t)$ is $\mathbb{E}\left( {s}(t) \right)    \equiv 0$.

Since  $ \phi(t+ NMT_s, \tau) = \phi(t, \tau) $ and  $ \mathbb{E}\left( {s}(t + NMT_s) \right) = \mathbb{E}\left( {s}(t) \right) $, $s(t)$ is \emph{cyclostationary} with a period $NMT_s$.

In order to compute the PSD of a cyclostationary process, the time-averaged autocorrelation function is obtained by averaging $\phi(t, \tau)$ over a single period $NMT_s$, i.e.,
\begin{align}
& \overline{\phi}(\tau)  = \frac{1}{NMT_s} \int_{-\frac{NMT_s}{2}}^{\frac{NMT_s}{2}} \phi(t, \tau) dt \notag \\
= &      \sum_{l=0}^{M-1}   \sum_{k=0}^{N-1}  \sum_{n= 0}^{N-1}  \sum_{\hat{n}= 0}^{N-1}  \frac{ \sigma_{l,k}^2 }{N}      e^{-\jmath 2\pi \frac{k (n - \hat{n})}{N} }    \phi_{gg}(\tau - (n-\hat{n})MT_s)  \notag
\end{align}
where
\begin{align}
&\phi_{gg}(\tau - (n-\hat{n})MT_s)  \notag \\
= &\frac{1}{NMT_s}   \int_{-\infty}^{\infty}  g_{I}( t )  g_{I}( t +\tau + (n-  \hat{n})MT_s)   dt
\end{align}

The Fourier transform of $\overline{\phi}(\tau)$ yields the PSD of $s(t)$ as follows
\begin{align}
& P_g(f) \notag \\
= & \sum_{l=0}^{M-1}   \sum_{k=0}^{N-1}   \frac{ \sigma_{l,k}^2 }{MN^2T_s}    \sum_{n= 0}^{N-1}  \sum_{\hat{n}= 0}^{N-1}  e^{- \jmath 2\pi \frac{k n - k \hat{n}}{N} }   \notag \\
&\mathcal{F}(\phi_{gg}(\tau - (n-\hat{n})MT_s))  \notag \\
= & \sum_{l=0}^{M-1}   \sum_{k=0}^{N-1}   \frac{ \sigma_{l,k}^2 }{MN^2T_s}  |G_I(f)|^2  \sum_{n= 0}^{N-1}  e^{-\jmath 2\pi (\frac{k  }{N} -   MT_s f)n } \cdot  \notag \\
&\sum_{\hat{n}= 0}^{N-1}  e^{\jmath 2\pi (\frac{  k }{N} - MT_s f)\hat{n}}    \notag \\
= &   \underbrace{\sum_{k=0}^{N-1}  \frac{\sigma_{k}^2}{T_s}    \left( \frac{ \sin(\pi (k -   MNT_s f) )   }{ N \sin(\frac{\pi}{N} (k -   MNT_s f) )  }   \right)^2}_{P_d(f)}  |G_I(f)|^2
\end{align}
where $ \sigma_{k}^2 = \frac{ \sum_{l=0}^{M-1} \sigma_{l,k}^2 }{M}$, $ G_I(f) = \mathcal{F}(g_I(t)) $ is the Fourier transform of the interpolation filter $g_I(t)$.
Note that $P_d(f)$ is the PSD of the discrete-time OTFS signal with the sampling rate $\frac{1}{T_s}$, and the derivation can be performed by simply replacing  $g_I( t)$ in \eqref{OTFSbaseband} by $\delta(t)$.

\end{appendices}

\bibliographystyle{IEEEtran}%
\bibliography{bibfile}

\end{document}